\documentclass[letterpaper,11pt]{article}
\usepackage{fullpage} 

\usepackage{url}
\usepackage{xspace}

\usepackage{amssymb,comment}
\usepackage{amsmath}
\usepackage{tikz}
\usepackage{pgflibrarysnakes}
\usetikzlibrary{snakes}
\usepackage{times}
\usepackage{amstext}
\usepackage{calc}
\usepackage{amsopn}
\usepackage[noend]{algorithmic}
\usepackage{algorithm}
\usepackage{eucal}
\usepackage{latexsym}
\usepackage{wrapfig}
\usepackage{tabularx}

\usepackage[margin=1in]{geometry}

\newcommand{\exactcover}{{{\sc Exact Cover by $3$-Sets}}\xspace}

\usepackage{amsthm}

\newtheorem{theorem}{Theorem}
\newtheorem{corollary}[theorem]{Corollary}

\newtheorem{lemma}[theorem]{Lemma}

\theoremstyle{definition}
\newtheorem{definition}[theorem]{Definition}

{\bf}{\rm}

\newcommand{\cP}{\ensuremath{\mathcal{P}}}
\newcommand{\cF}{\ensuremath{\mathcal{F}}}
\newcommand{\cS}{\ensuremath{\mathcal{S}}}
\newcommand{\Z}{\ensuremath{\mathbb{Z}}}

\newcommand{\dist}{\ensuremath{\mathrm{dist}}}
\newcommand{\radius}{\ensuremath{\mathrm{radius}}}
\newcommand{\nil}{\ensuremath{\mathrm{nil}}}
\newcommand{\eps}{\ensuremath{\epsilon}}

\newcommand{\algo}[1]{\texttt{#1}}

\floatname{algorithm}{Algorithm}

\newcommand{\defproblemapx}[4]{
  \vspace{1mm}
\noindent\fbox{
  \begin{minipage}{\textwidth}
  #1 \\ 
  {\bf{Input:}} #2  \\
  {\bf{Output:}} #3 \\
  {\bf{Goal:}} #4
  \end{minipage}
  }
  \vspace{1mm}
}

\newcommand{\defproblemu}[3]{
  \vspace{1mm}
\noindent\fbox{
  \begin{minipage}{\textwidth}
  #1 \\ 
  {\bf{Input:}} #2  \\
  {\bf{Question:}} #3
  \end{minipage}
  }
  \vspace{1mm}
}

\newcommand{\ra}{{\rightarrow}}

\begin{document}

\date{}

\author{  }
\title{LP Rounding for $k$-Centers with Non-uniform Hard Capacities\\
{\small (Extended Abstract)}}

\author{
Marek Cygan\thanks{
IDSIA, University of Lugano, Switzerland, \texttt{marek@idsia.ch}.
Research supported in part by ERC Starting Grant PAAl 259515 and Foundation for Polish Science.  }
\and
MohammadTaghi Hajiaghayi\thanks{University of Maryland, College Park MD 20742, \texttt{hajiagha@cs.umd.edu}. Research supported in part by NSF
CAREER award 1053605, ONR YIP award N000141110662, DARPA/AFRL award
FA8650-11-1-7162, and a University of Maryland Research and
Scholarship Award (RASA).}
\and
Samir Khuller\thanks{
Dept. of Computer Science, University of Maryland, College Park MD 20742, \texttt{samir@cs.umd.edu}.
Research supported by NSF CCF-0728839, NSF CCF-0937865 and a Google Research Award.} }

\begin{titlepage}
\def\thepage{}
\thispagestyle{empty}
\maketitle

\begin{abstract}
In this paper we consider a generalization of the classical $k$-center problem with 
capacities. Our goal is to select $k$ centers in a graph, and assign each node to 
a  nearby center, so that we respect the capacity constraints on centers. The objective
is to minimize the maximum distance a node has to travel to get to its assigned
center. This problem is $NP$-hard, even when centers have no capacity restrictions
and optimal factor 2 approximation algorithms are known.
With capacities, when all centers have identical capacities, a 6 approximation
is known with no better lower bounds than for the infinite capacity version. 

While many generalizations and variations  of this problem have been studied extensively,
no progress was made on the capacitated version for a general capacity function.
We develop the first constant factor approximation algorithm 
for this problem. Our algorithm  uses an LP rounding approach to solve this problem,
 and works for the case of non-uniform {\em hard} capacities, when 
multiple copies of a node may not be chosen and can be extended to the case
when there is a hard bound on the number of copies of a node that may be selected.
In addition we establish a lower bound on the integrality gap of 7(5) for 
non-uniform (uniform) hard capacities. 
In addition we prove that if there is a $(3-\epsilon)$-factor approximation for this problem
then $P=NP$.  

Finally, for non-uniform soft capacities we present a much simpler $11$-approximation
algorithm, which we find as one more evidence that hard capacities are much harder to deal with.
\end{abstract}

\end{titlepage}

\newpage

\section{Introduction}
The $k$-center problem is a classical facility location problem and
is defined as follows: given an edge-weighted graph $G = (V,E)$ 
find a subset $S \subseteq V$ of size at most $k$ such that each vertex
in $V$ is ``close'' to some vertex in $S$. More formally, once we choose $S$ 
the objective function is $ \max_{u \in V} \min_{v \in S} d(u,v),$
where $d$ is the distance function (a metric).
The problem is known to be NP-hard \cite{GJ}.
Approximation algorithms for the 
$k$-center problem have been well studied and are
known to be optimal \cite{Gonzalez,HS1,HS2,HN}. 
In this paper we consider the $k$-center problem with
{\em non-uniform capacities}. We have a 
capacity function $L$ defined for each vertex, hence
$L(u)$ denotes the capacity of vertex $u$. The goal is to 
identify a set $S$ of at most $k$ centers, as well as an {\em assignment} of 
vertices to ``nearby'' centers. No more than $L(u)$ vertices may be assigned to 
a chosen center at vertex $u$. Under these
constraints we wish to minimize the maximum distance between a vertex $v$
and its assigned center $\phi(v)$. Formally, the cost of a solution $S$ is
$ \min_{S \subseteq V, |S|=k} \max_{v \in V} d(v,\phi(v))$
such that 
$\left|\left\{v\mid \phi(v) = u\right\}\right| \leq L(u) \ \ \forall u \in S 
\mbox{ where }  \phi: V \rightarrow S. $

For the special case when all the capacities are {\em identical},
a 6 approximation was developed by Khuller and Sussmann \cite{KS}
improving the previous bound of 10 by Bar-Ilan, Kortsarz and Peleg \cite{BKP}.
In the special case when multiple copies of the same vertex may be chosen,  the
approximation factor was improved to 5. 
No improvements have been obtained on these results in the last 15 years.
The assumption that the
capacities are identical is crucial for both these approaches as it allows one 
to select centers and then  ``shift'' to a neighboring vertex.
In addition, one can use arguments such as $\lceil \frac{N}{L} \rceil$ is a
lower bound on the optimal solution; with non-uniform capacities we cannot
use such a bound. This problem has resisted any progress at all, and 
no constant approximation algorithm was developed for the non-uniform capacity version.
 
In this work we present the first constant factor approximations for 
the $k$-center problem with arbitrary capacities. Moreover, our algorithm
satisfies hard capacity constraints and only one copy of any vertex is chosen. 
When multiple copies of a vertex can be chosen then a constant
factor approximation is implied by our result for the hard capacity
version. For convenience, we discuss the algorithm for the case when
at most one copy of a vertex may be chosen.
Our algorithms use a novel LP rounding method to obtain the result.
In fact this is the first time that LP techniques have been applied 
for any variation of the $k$-center problem.


While our constants are large, we do show via integrality
gap examples that the problem with non-uniform capacities is significantly 
harder than the basic $k$-center problem. 
In addition we establish that if there is a 
$(3-\epsilon)$-approximation for the $k$-center problem with non-uniform capacity constraints then
$P=NP$. Such a result is known for the cost $k$-center problem \cite{JACM} and from that
one can infer the result for the unit cost capacitated $k$-center problem with
non-uniform capacities, but our reduction is a direct reduction from Exact Cover by $3$-Sets and
considerably simpler.
We would like to note that for the $k$-supplier problem, which can be seen as $k$-center
with disjoint sets of clients and potential centers, 
a simple proof of $(3-\eps)$ approximation hardness under $P\neq NP$
was obtained by Karloff and can be found in~\cite{HS2}.

In all cases of studying covering problems, the hard capacity restriction
makes the problems very challenging. For example, for the simple
capacitated vertex cover problem with soft capacities, a 2 approximation can be obtained by
a variety of methods \cite{GHKO,GKPS} -- however imposing a hard capacity restriction
makes the problem as hard as set cover \cite{CN}. In the special
case of unweighted graphs it was shown
that a 3 approximation is possible \cite{CN},
which was subsequently improved to 2 \cite{GHKKS}.

\vspace*{-0.3cm}
\subsection{Related Facility Location Work}

The {\em facility location} problem is a central problem in operations
research and computer science and has been a testbed for many new
algorithmic ideas resulting a number of different
approximation algorithms. 
In this problem,  given a metric (via a weighted graph $G$), a set of nodes called {\em clients},  and opening costs on some nodes called {\em facilities}, the goal is to open a subset of facilities such that the sum of their opening costs and connection costs of clients to their nearest open facilities is minimized. 
When the facilities have capacities, the problem is called the {\em capacitated facility location} problem.
The first constant-factor approximation algorithm for the
(uncapacitated) version of this problem  was given by Shmoys, Tardos,
and Aardal~\cite{STA} and was based on LP rounding and a filtering
technique due to Lin and Vitter~\cite{LV}. 
A long series of improvements culminated in a $1.5$ approximation
due to Byrka \cite{Byr07}.
Up to now, the best known approximation ratio is
1.488, due to Li~\cite{Li11} who uses a randomized selection in
Byrka's algorithm \cite{Byr07}. 
Guha and Khuller~\cite{GK} showed that this problem is hard to
approximate within a factor better than 1.463, assuming $NP
\not\subseteq DTIME[n^{O(\log \log n)}]$.

Capacitated facility location has also received a great deal of
attention in recent years. Two main variants of the problem are
soft-capacitated facility location and hard-capacitated facility
location: in the latter problem, each facility is either opened at
some location or not, whereas in the former, one may specify any
integer number of facilities to be opened at that location. Soft
capacities make the problem easier and by modifying approximation
algorithms for the uncapacitated problems, we can also handle this
case~\cite{STA,JV}.
 Korupolu, Plaxton, and Rajaraman~\cite{KPR} gave the first constant-factor approximation
algorithm that handles hard capacities, based on a local search
procedure, but their approach works only if all capacities are
equal. Chudak and Williamson~\cite{CW05} improved this performance
guarantee to 5.83 for the same uniform capacity case. P{\'a}l,
Tardos, and Wexler~\cite{PTW01} gave the first constant performance
guarantee for the case of non-uniform hard capacities. This was
recently improved by Mahdian and P{\'a}l~\cite{MP03} and Zhang,
Chen, and Ye~\cite{ZCY04} to yield a 5.83-approximation algorithm.
All these approaches are based on local search. The only
LP-relaxation based approach for this problem is due to Levi, Shmoys
and Swamy~\cite{LSS04} who gave a 5-approximation algorithm for the
special case in which all facility opening costs are equal
(otherwise the LP does not have a constant integrality gap).
The above approximation algorithms for hard capacities are focused on
the uniform demand case or the splittable case in which each unit of
demand can be served by a different facility. Recently, Bateni and
Hajiaghayi~\cite{BH09} considered the unsplittable hard-capacitated
facility location problem when we allow violating facility
capacities by a $1+ \epsilon$ factor (otherwise, it is NP-hard to
obtain any approximation factor) and obtain an $O(\log n)$
approximation algorithm for this problem.

A problem very close to both facility location and $k$-center is the {\em $k$-median} problem in which we want to open at most $k$ facilities 
(like in the $k$-center problem) 
and the goal is to minimize the sum  of connection costs of  clients to their nearest open facilities (like facility location). If facilities have capacities the problem is called {\em capacitated $k$-median}.  The approaches for uncapacitated facility location often work for $k$-median. In particular,  Charikar, Guha, Tardos, and Shmoys~\cite{CGTS}  
gave  the first constant factor approximation for 
$k$-median  based on LP rounding.
The best approximation factor for $k$-median is $3+\epsilon$, for an arbitrary positive constant $\epsilon$, via the local search algorithm of Arya et al.~\cite{kmed3}. Unfortunately obtaining a constant factor approximation algorithm for capacitated $k$-median still remains open despite consistent effort. The methods used to solve uncapacitated $k$-median or even the local search technique for capacitated facility location all seem to suffer from serious drawbacks when trying to apply them for capacitated $k$-median.
For example standard LP relaxation is known to have an unbounded integrality gap~\cite{CGTS}. The only previous attempts with  constant approximation factors for this problem violate the capacities within a constant factor for the uniform capacity case~\cite{CGTS} and the non-uniform capacity case~\cite{CR05} or exceed the number $k$ of facilities by a constant factor~\cite{BCR01}. 

\noindent
{\bf Removing the metric:}
We employ the standard ``thresholding'' method used for bottleneck 
optimization problems. We can assume that we guess the optimal
solution, since there are polynomially many distinct distances between
pairs of nodes. Once we guess the distance correctly, we create an
unweighted graph consisting of those edges $uv$ such that 
$d(u,v) \le OPT$. 
We henceforth assume that we are considering the 
problem for an undirected graph $G$.

\defproblemapx{{\sc Capacitated $k$-Center Problem}}
{An undirected graph $G=(V,E)$, a capacity function $L:V \rightarrow \mathbb{N}$ and an integer $k$.}
{A set $S \subseteq V$ of size $k$, and a function $\phi:V\rightarrow S$, such that 
  for each $u \in S$, $|\phi^{-1}(u)| \le L(u)$.}
{Minimize $\max_{v \in V} \dist_G(v, \phi(v))$.}

By a $c$-approximation algorithm we denote a polynomial time
algorithm, that for an instance for which there exists
a solution with objective function equal to $1$,
returns a solution using distances at most $c$.
Note that the distance function $\dist(u,v)$, measures the distance in the unweighted undirected graph.

In the soft-capacitated version $S$ can be a multiset, 
that is one can open more than one center at a vertex.
To avoid confusion we call the standard version of the problem hard-capacitated.

\vspace*{-0.3cm}
\subsection{Our results}


While LP based algorithms have been widely used for uncapacitated facility location
problems as well as capacitated versions of facility location with soft
capacities, these methods are not of much use for problems in
dealing with hard capacities due to the fact that they usually have an
unbounded integrality gap \cite{CGTS,PTW01}. 

For general undirected graphs this is also the case for the capacitated $k$-center problem.
Consider the LP relaxation for the natural IP, which we denote as LP1.
We use $y_u$ as an indicator variable for open centers.

\begin{align}
\label{con1}    & \textstyle{\sum_{u \in V}y_{u} = k;} &  & \\
\label{con2}    & \textstyle{x_{u,v} \le y_u} & \textstyle{\forall u,v \in V} & \\
\label{con3}    & \textstyle{\sum_{v \in V} x_{u,v} \le L(u)y_u} & \textstyle{\forall u \in V} & \\
\label{con4}    & \textstyle{\sum_{u \in V} x_{u,v} = 1} & \textstyle{\forall v \in V} & \\
\label{con5}    & \textstyle{0 \le y_u \le 1} & \textstyle{\forall u\in V} & \\
\label{con6}    & \textstyle{x_{u,v} = 0} & \textstyle{\forall u,v \in V\ \dist_G(u,v) > 1} &  \\
\label{con7}    & \textstyle{x_{u,v} \ge 0} & \textstyle{\forall u,v \in V} & 
    \end{align}

For the sake of presentation we have introduced variables $x_{u,v}$ 
for all $u$, $v$, even if the distance between $u$ and $v$ in $G$ is greater than one. 
We will use those variables in our rounding algorithm.
Furthermore in constraints (\ref{con1}) and (\ref{con4}) we used equality
instead of inequality to make our rounding algorithm
and lemma formulations simpler.
In the soft-capacitated version the $y_u \le 1$ part of constraint~(\ref{con5}) should be removed.
Note that we are only interested in feasilibity of LP1, and there is no objective function.

For an undirected graph $G=(V,E)$ and a positive integer $\delta$, by $G^\delta$ we denote the graph $(V,E')$,
where $uv \in E'$ iff $\dist_G(u,v) \le \delta$.
By an integrality gap of LP1 we mean the minimum positive integer $\delta$
such that if LP1 has a feasible solution, then the graph $G^\delta$ 
admits a capacitated $k$-center solution.
As this is usually the case for capacitated problems,
by a simple example we prove LP1 has unbounded integrality gap for general graphs.

\begin{theorem}
\label{thm:gap-unbounded}
LP1 has unbounded integrality gap, even for uniform capacities.
\end{theorem}

\begin{proof}
Let $G'$ be a graph that consists of two adjacent vertices $a$, $b$ together with $4$ vertices adjacent to both of them (see Fig.~\ref{fig1}).
Set uniform capacity $L=4$, $k=3$ and consider the graph $G$ which is a disjoint union of two copies of $G'$.
Observe that by setting $y_{a} = y_{b} = 0.75$ in each of the copies
as well as $x_{a,v} = x_{b,v} = 0.5$, for all six vertices $v$, we obtain a feasible solution to the LP relaxation.
No matter what $\delta$ we choose, there is no capacitated $k$-center with $L=4$ in the graph $G^\delta$.
\end{proof}

\begin{figure}[h]
\begin{center}
\includegraphics{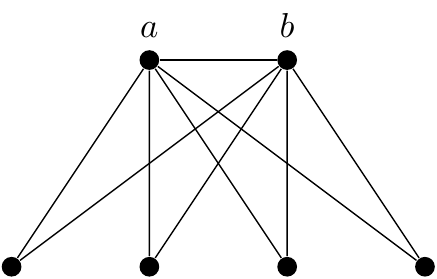}
\caption{A single connected component of the graph $G$ used in the proof of Theorem~\ref{thm:gap-unbounded}.
}
\label{fig1}
\end{center}
\end{figure}

However, interestingly, if we assume that the given graph is connected, 
the situation changes dramatically.
Our main result is, that both for hard and soft capacitated
version of the $k$-center problem, even for non-uniform capacities, LP1 has constant integrality gap
for connected graphs.
Moreover by using novel techniques we show a corresponding
polynomial time rounding algorithm, which consists
of several steps, described at high level in the following subsection.
The actual algorithm is somewhat complex, although it can be implemented
quite efficiently.

\begin{theorem}
\label{thm:main}
There is a polynomial time algorithm, which
given an instance of the hard-capacitated $k$-center problem
for a connected graph, and a fractional feasible solution for LP1,
can round it to an integral solution that uses non-zero $x_{u,v}$
variables for pairs of nodes with distance at most $c$.
\end{theorem}

\begin{corollary}
The integrality gap of LP1 for connected graphs is bounded by a constant,
and there is a constant factor approximation algorithm for connected graphs.
\end{corollary}

To simplify the presentation we do not calculate the exact constant proved
in the above corollary, but it is in the order of hundreds.
As a counterposition, for soft capacities in Section~\ref{sec:soft} we present a much simpler
$11$-approximation algorithm, which we find as one more evidence
that hard capacities are much harder to deal with.

\begin{theorem}
\label{thm:soft-apx}
For connected graphs there is a polynomial time rounding algorithm,
upper bounding the integrality gap of LP1 by $11$ for soft-capacities.
\end{theorem}

By using standard techniques one can restrict the capacitated $k$-center
problem to connected graphs.

\begin{theorem}
\label{thm:conn}
If there exists a polynomial time $c$-approximation algorithm for the (soft) capacitated $k$-center problem in connected graphs,
then there exists a polynomial time $c$-approximation algorithm for general graphs.
\end{theorem}

\begin{proof}
Let us assume that we are given a graph $G$ with connected components $C_1,\ldots,C_t$, a capacity function $L:V \ra \mathbb{N}$ and an integer $k$.
For each connected component $C_i$ using binary search we find the smallest value of $k_i \le k$ for which our black box algorithm
finds a solution (if there is no solution for $k_i=k$ then we set $k_i = \infty$).
If $\sum_{i=1}^t k_i > k$ then we answer NO, otherwise we return the union of solutions found by the black box algorithm.

To prove correctness of the above algorithm it is enough to observe that if there exists a solution $S$ to the $k$-center problem in the graph $G$ then
for each connected component $C_i$ we have $k_i \le |S\cap V(C_i)|$.
\end{proof}


Therefore we prove there is a constant factor approximation algorithm for the hard-capacitated $k$-center problem\footnote{With some care,
perhaps some of the constants can be improved, however our focus was to
show that a constant approximation is obtainable using LP rounding.}.
Our results easily extend to the case when there is an upper bound $U(u)$ of the number of times vertex
$u$ may be chosen as a center.  Constraint~\ref{con5} should be modified to be $0 \le y_u \le U(u)$ 
to yield a relaxation LP2. We
can employ the same rounding procedure as discussed for the hard capacity case with $U(u)=1$.

The proof of the following theorem is omitted.
\begin{theorem}
\label{thm:main2}
There is a polynomial time algorithm, which
given an instance of the hard-capacitated $k$-center problem
for a connected graph, and a fractional feasible solution for LP2, 
can round it to an integral solution that uses non-zero $x_{u,v}$
variables for pairs of nodes with distance at most $c$. 
\end{theorem}

While our constants are large, we do show in Section~\ref{sec:gaps} via integrality
gap examples that the problem with non-uniform capacities is significantly 
harder than the basic $k$-center problem. 

\begin{theorem}
\label{thm:gap-constant}
For connected graphs the integrality gap of LP1 is at least $5$
for uniform-hard-capacities and at least $4$ for uniform-soft-capacities.

Moreover in the non-uniform hard-capacitated case, the integrality gap of LP1 for connected graphs is at least $7$, even if all the non-zero capacities are equal.
\end{theorem}

Despite the fact, that the algorithm of~\cite{KS} for uniform capacities was obtained more than a decade ago,
no lower bound for the capacity version (neither soft nor hard), better than the trivial $2-\eps$, derived from the uncapacitated version, is known.
We believe that the integrality gap examples, presented in this paper, are of independent 
interest since they may help in proving a stronger lower bound for the capacitated $k$-center problem with uniform capacities.

To make a step in this direction we investigate lower bounds for the
non-uniform case.
By a reduction from the cost $k$-center problem~\cite{JACM}
one can show that there is no $(3-\epsilon)$-approximation 
for the capacitated $k$-center problem with non-uniform capacities.
By a simple reduction from Exact Cover by $3$-Sets, in 
Section~\ref{sec:lowerbound},
we prove the same result under the assumption $P\not=NP$.

Finally we give  evidence that our LP approach might be the proper tool
for solving the capacitated $k$-center problem. The proof of the 
following theorem shows that when the Khuller-Sussmann algorithm
fails to find a solution then in fact there is no feasible LP solution
for that guess of distance. The smallest radius guess for which the 
algorithm succeeds, proves an integrality gap on the LP.
Considering the result of Theorem~\ref{thm:gap-constant}, 
as we show in Section~\ref{sec:uniform}, it follows that for uniform capacities the gap in the analysis is small,
since our bounds are tight up to an additive $+1$ error.

\begin{theorem}
\label{thm:gap-upper-constant}
For connected graphs the integrality gap of LP1 is at most $6$
for uniform-hard-capacities and at most $5$ for uniform-soft-capacities.
\end{theorem}

\vspace*{-0.3cm}
\subsection{Our techniques}

We assume that $G$ is connected and that LP1 has a feasible solution 
for the graph $G$.
We call two functions $x:V\times V \ra {\mathbb R}_+ \cup \{0\}$ and $y:V \ra {\mathbb R}_+ \cup \{0\}$
an {\em assignment} even if $(x,y)$ is potentially infeasible for LP1.
In other words initially we have a feasible fractional solution, 
in the end we will obtain a feasible integral solution,
although during the execution of our rounding algorithm an assignment 
$(x,y)$ is not required to be feasible.
Furthermore without loss of generality we assume that for a vertex $v$ with $L(v)=0$ we have $y_v=0$.

We need to show that there exists a constant $\delta$ such that if for a connected component
LP1 has a feasible solution, then one can (in polynomial time) find an
integral feasible solution for $G^\delta$.

\begin{definition}[{\bf $\delta$-feasible solution}]
An assignment is called {\em $\delta$-feasible} if it is feasible for the graph $G^\delta$.
\end{definition}

Note that the only difference between LP1's for the graphs $G$ and $G^\delta$ is constraint (\ref{con6}).


\begin{definition}[{\bf radius$_{(x,y)}$}]
For a $\delta$-feasible solution $(x,y)$ to LP1 we define a function $\radius_{(x,y)} : V \ra \{0, \ldots, \delta\}$
which for a vertex $u$ assigns the greatest integer $i$ such that there exists a vertex $v$ with $\dist_G(v,u) = i$ and $x_{u,v} > 0$
(if no such $i$ exists then $\radius_{(x,y)}(u)=0$).
\end{definition}

We give a brief overview of the rounding algorithm described in subsequent
subsection of Section~\ref{sec:rounding}.
Initially we start with a $1$-feasible (fractional) solution $(x,y)$ to LP1
and our goal is to make it integral. 
We perform several steps where in each step we get more structure
on the $\delta$-feasible solution but at the same time the value
of $\delta$ will increase. 

In Sections~\ref{sec:path-structure}-\ref{sec:rounding-flow}
in four non-trivial steps we round the $y$-values of a feasible solution.
First, in Section~\ref{sec:path-structure} 
we define a {\em caterpillar structure} which is a key structure in the
rounding process.
We show that in polynomial time we can find a $5$-feasible solution 
together with a caterpillar structure $(P,P')$ such that all vertices
outside of the caterpillar structure have integral $y$-values.
In Section~\ref{sec:chain-shifting} we define the {\em $y$-flow} and {\em chain shifting}
operations which allow for transferring $y$-values between distant vertices
using intermediate vertices on the caterpillar structure.
Unfortunately, because the capacities are non-uniform and hard,
to find a rounding flow for a caterpillar structure we need more
assumptions. 
To overcome this difficulty in the most challenging part
of the rounding process, that is in Section~\ref{sec:separable},
we define a {\em safe} caterpillar structure and show 
how to split a given caterpillar structure into a set of safe caterpillar structures
(at the cost of increasing radius of the $\delta$-feasible solution).
In Section~\ref{sec:rounding-flow} we design a rounding procedure
for a safe caterpillar structure, obtaining a $c$-feasible solution
with integral $y$-values, for some constant $c$. 
We would like to note, that for uniform capacities every caterpillar structure
is safe, therefore for non-uniform capacities we have to design much more
involved tools comparing to the previously known uniform capacities case.

Finally in Section~\ref{sec:rounding-x} we show, that using standard techniques,
when we have integral $y$-values then rounding $x$-values is simple, 
obtaining a constant factor approximation algorithm.




\section{LP rounding for hard-capacities}
\label{sec:rounding}

\vspace*{-0.3cm}
\subsection{Group shifting and caterpillar structure}
\label{sec:path-structure}

In the first phase of our procedure we obtain a path-like structure
containing all vertices with non-integral $y$-values.
We first define the notion of shifting values
between variables of LP1 relaxation.

\begin{definition}[\bf shifting]
For an assignment $(x,y)$ for the $LP$, two distinct vertices $a,b \in V$ 
and a positive real $\alpha \le \min(y_a, 1-y_b)$ such that $L(a) \le L(b)$
by {\em shifting} $\alpha$ from $a$ to $b$
we consider the following operation:
\begin{enumerate}
  \item Let $\eps = \frac{\alpha}{y_a}$;
for each $v \in V$ let $\Delta_v=\eps x_{a,v}$, 
decrease $x_{a,v}$ by $\Delta_v$ and increase $x_{b,v}$ by $\Delta_v$.
\item Increase $y_b$ by $\alpha$, and decrease $y_a$ by $\alpha$.
\end{enumerate}
\end{definition}

\begin{lemma}
\label{lem:shift}
Let $(x,y)$ be a $\delta$-feasible solution to $LP$.
Let $(x',y')$ be a result of shifting $\alpha$ from $a$ to $b$, for some $\alpha,a,b$ such that $L(a) \le L(b)$, $0 < \alpha \le \min(y_a,1-y_b)$.
Then $(x',y')$ is a $(\delta+\dist_G(a,b))$-feasible solution 
and for each vertex $v \not= b$ we have $\radius_{(x',y')}(v) \le \radius_{(x,y)}(v)$
whereas $\radius_{(x',y')}(b) \le \max(\radius_{(x,y)}(a)+\dist_G(a,b),\radius_{(x,y)}(b))$.
\end{lemma}

\begin{proof}
First we prove that $(x',y')$ is a $(\delta+\dist_G(a,b))$-feasible solution.
Since the sum of all $y$-values in $(x,y)$ is equal to the sum of all $y$-values in $(x',y')$, constraint $(\ref{con1})$ of LP1 is satisfied.
To prove $(\ref{con2})$ of LP1, it is enough to consider the variables $x'_{a,v},x'_{b,v}$ for each vertex $v$,
that is:

\begin{align*}
x'_{b,v} & = x_{b,v}+\eps x_{a,v} \le y_b + \eps y_a = y_b + \alpha = y_b' \\
x'_{a,v} & = x_{a,v}-\eps x_{a,v} \le y_a(1-\eps) = y_a' \,.
\end{align*}

For $(\ref{con3})$ of LP1, we only verify $v=a$ and $v=b$ since for other vertices the sum did not change.
\begin{align*}
\sum_{v \in V} x'_{b,v} & = \sum_{v \in V} (x_{b,v}+\eps x_{a,v}) = (\sum_{v \in V}x_{b,v})+\eps(\sum_{v \in V} x_{a,v}) \le L(b)y_b+\eps L(a)y_a \le L(b)y_b+\eps L(b)y_a \\
   & = L(b)(y_b + \eps y_a) = L(b)y_b' \\
\sum_{v \in V} x'_{a,v} & = \sum_{v \in V} (x_{a,v}-\eps x_{a,v}) = (1-\eps) \sum_{v \in V}x_{a,v} \le (1-\eps) L(a)y_a = L(a)y_a'
\end{align*}
For each vertex $u$ the sum $\sum_{v \in V} x_{v,u}$ is equal to $\sum_{v \in V} x_{v,u}'$ hence constraint $(\ref{con4})$ is satisfied.
Constraints (\ref{con5}), (\ref{con7}) may be checked directly, since $\alpha \le \min(y_a, 1-y_b)$.

Since for each vertex $v \not= b$ a variable $x_{v,u}$ can only be decreased (when $v=a$),
therefore $\radius_{(x',y')}(v) \le \radius_{(x,y)}(v)$.
For $v=b$ the new radius may increase but it will not exceed $\radius_{(x,y)}(a)+\dist_G(a,b)$
since if $x'_{b,u} > 0$ then either $x_{b,u}>0$ or $x_{a,u} > 0$.
Consequently constraint (\ref{con6}) is satisfied for $G^{\delta'}$ where $\delta' = \delta+\dist_G(a,b)$ and the lemma follows.
\end{proof}

\begin{definition}[\bf group shifting]
For a $\delta$-feasible solution $(x,y)$ and a set $V_0 \subseteq V$ by a {\em group shifting}
we denote the following operation. Assume $V_0 = \{v_1,\ldots,v_{\ell}\}$, where $L(v_i) \le L(v_{i+1})$ for $1 \le i < \ell$.
As long as there are at least two vertices in $V_0$ with fractional $y$-values,
let $a$ be the smallest, and $b$ the greatest integer such that $v_a, v_b \in V_0$
are vertices with fractional $y$-values.
Shift $\min(y_a,1-y_b)$ from $a$ to $b$.
\end{definition}

\begin{lemma}
\label{lem:group-shift}
Let $(x,y)$ be a $\delta$-feasible solution, $V_0$ be a subset of $V$
and $d=\max_{a,b \in V_0} \dist_G(a,b)$.
After group shifting on $V_0$ we obtain a $(\delta+d)$-feasible solution $(x',y')$,
where there is at most one vertex in $V_0$ with fractional $y$-value
and moreover for $v \in V\setminus V_0$ we have $\radius_{(x',y')}(v) \le \radius_{(x,y)}(v)$.
\end{lemma}
To make a graph Hamiltonian we use the following lemma
known from 1960~\cite{karaganis, sekanina}. 

\begin{lemma}
\label{lem:ham}
For any undirected connected graph $G$ there always exists a Hamiltonian path in $G^3$
and one can find it in polynomial time.
\end{lemma}

We define a caterpillar structure which is one of the key ingredients 
of our rounding process.
Intuitively we want to define an auxiliary path-like tree, 
where adjacent vertices are close in the original graph $G$,
vertices with fractional $y$-values are leaves of the tree,
and all non-leaf vertices have $y$-values equal to $1$.
\begin{definition}[\bf caterpillar structure]
\label{def:path-structure}
By a $\delta$-caterpillar structure for an assignment $(x,y)$ we denote a sequence of distinct vertices $P=(v_1,\ldots,v_p)$
together with a sequence $P'=(v_0',\ldots,v_{p+1}')$ where:
\begin{enumerate}
  \item for each $i=1,\ldots,p$ we have $y_{v_i}=1$,
  \item for each $i=1,\ldots,p-1$ we have $\dist_G(v_i,v_{i+1}) \le \delta$,
  \item for each $i=0,\ldots,p+1$ either $v_i'=\nil$ or $v_i' \in V \setminus \{v_j : j=1,\ldots,p\}$,
  \item\label{def:path:point4} for each $i=1,\ldots,p$ if $v_i'\not=\nil$ then $L(v_i) \ge L(v_i')$, $0 < y_{v_i'} < 1$, $\dist_G(v_i,v_i') \le \delta$,
  \item if $v_0' \not= \nil$ then $\dist_G(v_0',v_1) \le \delta$, $0 < y_{v_0'} < 1$,
  \item if $v_{p+1}' \not= \nil$ then $\dist_G(v_{p+1}',v_p) \le \delta$, $0 < y_{v_{p+1}'} < 1$,
  \item for each $0 \le i < j \le p+1$ if $v_i'\not=\nil$ and $v_j'\not=\nil$ then $v_i'\not=v_j'$,
  \item\label{def:path:point8} $\sum_{v \in V(P')} y_v$ is integral.
\end{enumerate}
\end{definition}

We sometimes omit $\delta$ and simply write ``caterpillar structure'' 
when the value of $\delta$ is irrelevant.

\begin{figure}[h]
\begin{center}
\includegraphics{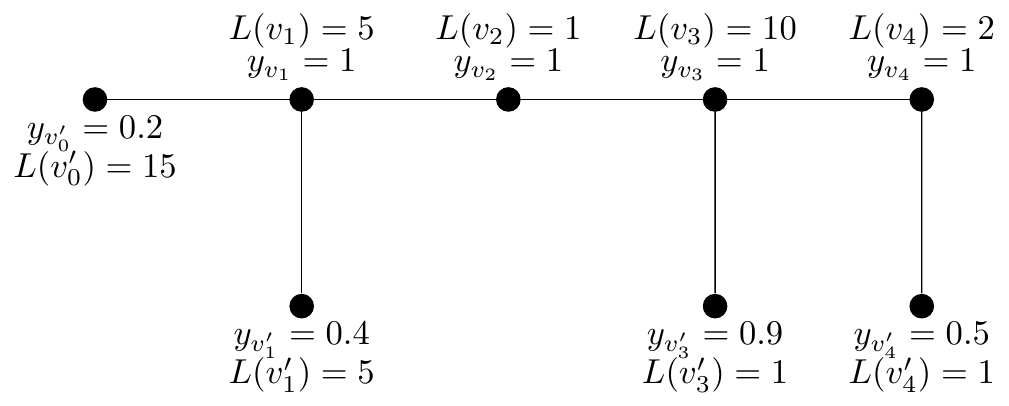}
\caption{Example of a $\delta$-caterpillar structure $((v_1,v_2,v_3,v_4),(v_0',v_1',\nil,v_3',v_4',\nil))$.
  Vertices connected by edges are within distance $\delta$ in the graph $G$.
    Note that the sum of $y$-values over all vertices is integral.
}
\label{fig2}
\end{center}
\end{figure}

\begin{lemma}
\label{lem:path-structure}
For a given feasible LP solution $(x,y)$ we can find a $5$-feasible solution $(x',y')$
together with a $21$-caterpillar structure $(P,P')$ such that each vertex $v \in V \setminus (V(P) \cup V(P'))$ 
has an integral $y$-value in $(x',y')$, and the first and last element of the sequence $P'$ equals $\nil$.
\end{lemma}

\begin{proof}
Consider the following algorithm for constructing sets $S$, $S'$ and a function $\Phi: V \ra S'$.
The set $S$ will be an inclusionwise maximal independent set in $G^2$ and moreover we ensure that $L(\Phi(v)) \ge L(v)$, for any $v \in V$.

\begin{enumerate}
  \item Set $V_0 := V$ and $S := S' := \emptyset$.
  \item As long as $V_0 \not= \emptyset$ let $v$ be a highest capacity vertex in $V_0$.
    \begin{itemize}
      \item Let $f(v)$ be a highest capacity vertex in $N_G[v]$ (potentially $f(v) \not\in V_0$).
      \item Add $f(v)$ to $S'$ and for each $u \in N_G[N_G[v]] \cap V_0$ set $\Phi(u)=f(v)$.
      \item Add $v$ to $S$ and set $V_0 := V_0 \setminus N_G[N_G[v]]$.
    \end{itemize}
\end{enumerate}

Observe that each time we remove from the set $V_0$ all vertices that are within distance two from $v$,
hence the set $S$ is an inclusion maximal independent set in $G^2$.
For this reason vertices in the set $S$ have disjoint neighborhoods
and moreover by constraints (\ref{con4}) and (\ref{con2}) of the LP1 we infer that for each $v \in V$ we have:

\begin{equation}
\label{eq1}
\sum_{u \in N[v]} y_u \ge \sum_{u \in N[v]} x_{u,v} = 1\,
\end{equation}
We perform shifting operations to make sure all vertices in the set $S'$ have $y$-value equal to one.
Consider a vertex $v \in S$ and the corresponding vertex $f(v)$ chosen by the algorithm.
As long as $y_{f(v)} < 1$ take any $u \in N[v], u \not= f(v)$ such that $y_u > 0$
and shift $\min(y_u,1-y_{f(v)})$ from $u$ to $f(v)$.
Note that $L(u) \le L(f(v))$ by the definition of $f(v)$ and for this reason shifting is possible.
By Lemma~\ref{lem:shift} after all the shifting operations we have a $3$-feasible solution $(x,y)$,
since before a shift from $u$ to $f(v)$ we have $\radius_{(x,y)}(u) \le 1$, $\radius_{(x,y)}(f(v)) \le 3$ and $\dist_G(u,f(v))\le 2$. 
Moreover by Inequality (\ref{eq1}) we infer,
that all the vertices in the set $S'$ have $y$-value equal to one,
     since otherwise a shifting operation from some $u \in N[v]$ to $f(v)$ would be possible.

Observe that by the maximality of the independent set $S$ in $G^2$ the graph $G^5[S]$ is connected,
otherwise we could add a vertex to $S$ still obtaining an independent set in $G^2$.
Moreover for any two adjacent vertices $u,v \in S$ in $G^5[S]$,
the vertices $f(u),f(v)$ are adjacent in $G^7[S']$.
By the connectivity of $G^5[S]$, the graph $G^7[S']$ is also connected.
By Lemma~\ref{lem:ham} we can in polynomial time order the vertices of $S'$ to obtain a Hamiltonian path $P$ in $G^{21}[S']$.

Currently for each vertex $v$ from the set $V \setminus S'$ we have $\radius_{(x,y)}(v) \le 1$.
For each $v\in S$ we use group shifting on the set $\Phi^{-1}(f(v)) \setminus S'$.
Since
\begin{align*}
\max_{a,b \in \Phi^{-1}(f(v)) \setminus S'}\dist_G(a,b) \le \max_{a,b \in \Phi^{-1}(f(v)) \setminus S'} \dist_G(a,v)+\dist_G(v,b) \le 4\,,
\end{align*}
by Lemma~\ref{lem:group-shift} we obtain a $5$-feasible solution $(x,y)$ 
such that all vertices in the set $S'$ have $y$-value equal to one 
and moreover for each $f(v) \in S'$ the set $\Phi^{-1}(f(v)) \setminus S'$ 
contains at most one vertex with fractional $y$-value.
Let us assume that the already constructed path $P$ is of the form $P=(v_1,\ldots,v_p)$.
We construct a sequence $P'=(\nil,v_1',\ldots,v_p',\nil)$ where as $v_i'$ we take the only vertex from $\Phi^{-1}(v_i) \setminus S'$
that has fractional $y$-value, or we set $v_i':=\nil$ if $\Phi^{-1}(v_i) \setminus S'$ has no vertices with fractional $y$-value.
Note that since the way we select vertices to the sets $S,S'$ is capacity driven (recall as $v$ we select
the highest capacity vertex in $V_0$ and as $f(v)$ we select a highest capacity vertex in $N[v]$),
for each vertex $u \in \Phi^{-1}(v_i)$ we have $L(u) \le L(v_i)$.
In this way we have constructed a $5$-feasible solution $(x,y)$ together with a desired $21$-caterpillar structure $(P,P')$.
\end{proof}

As the reader might notice in the above proof we always construct a caterpillar structure with $v_0'=v_{p+1}'=\nil$.
The reason why the definition of a caterpillar structure allows for $v_0'$ and $v_{p+1}'$ have non-$\nil$ values
is that in Section~\ref{sec:separable} we will split a caterpillar structure into two smaller pieces
and in order to have those pieces satisfy Definition~\ref{def:path-structure} we need $v_0'$ and $v_{p+1}'$.

\vspace*{-0.3cm}
\subsection{$y$-flow and chain shifting}
\label{sec:chain-shifting}

In the previous section we defined a group shifting operation.
Unfortunately we can only perform such an operation if vertices are close.
In this section we define notions of {\em$y$-flow} and {\em chain shifting}
which allow us to transfer $y$-value between distant vertices.
We will use those tools in Sections~\ref{sec:separable} and \ref{sec:rounding-flow}.

\begin{definition}[\bf $y$-flow]
For a given assignment $(x,y)$ let $S \subseteq V$ and $T \subseteq V$ be two disjoint sets
and let $\cF$ be a set containing sequences of the form $(\alpha,v_1,\ldots,v_t)$ representing paths, 
where $\alpha$ is a positive real, each $v_i \in V$ is a vertex (for $i=1,\ldots,t$), $v_1 \in S$, $v_t \in T$, $L(v_1) \le L(v_t)$ and for $i=2,\ldots,t-1$ we have $v_i \not\in S \cup T, y_{v_i} = 1, L(v_i) \ge L(v_1)$.
We call $(\alpha,v_1,\ldots,v_t)$ a {\em path} transferring $\alpha$ from $v_1$ to $v_t$ through $v_2,\ldots,v_{t-1}$.
We denote $v_2,\ldots,v_{t-1}$ as {\em internal} vertices of the path  $(\alpha,v_1,\ldots,v_t)$.

The set $\cF$ is a {\em$y$-flow} from $S$ to $T$ iff:
\begin{itemize}
  \item for each $v \in S$ the sum of values transferred from $v$ in $\cF$ is at most $y_v$,
  \item for each $v \in T$ the sum of values transferred to $v$ in $\cF$ is at most $1-y_v$,
  \item for each $v \in V \setminus (S \cup T)$ the sum of values transferred through $v$ in $\cF$ is at most $1$.
\end{itemize}
\end{definition}

For a given $y$-flow $\cF$ from $S$ to $T$ we define $G_{\cF}=(V,A)$ 
as an auxiliary directed graph with the same vertex set as $G$,
where an arc $(u,v)$ belongs to $A$ iff there is a path
in $\cF$ containing $u$ and $v$ as consecutive vertices in exactly this order.
We call the $y$-flow $\cF$ {\em acyclic} iff the directed flow graph $G_{\cF}$ is acyclic.
Furthermore we define a function $f_{\cF} : A \ra (0,1]$, which for an arc $(u,v)$
assigns the sum of $\alpha$ values in all the paths in $\cF$ that contain $u$ 
and $v$ as consecutive vertices.
Moreover by $fl_{\cF} : A \ra \mathbb{R}_+$ we denote a function, which
for an arc $(u,v)$ assigns the sum of terms $L(s)\alpha$
over all paths from ${\cF}$ that start with $\alpha$ and $s \in S$
and contain $u,v$ as consecutive elements.
Intuitively by $f_{\cF}((u,v))$ we denote the fractional number of centers that are transferred
from $u$ to $v$, whereas by $fl_{\cF}((u,v))$ we denote the fractional number of vertices (clients)
that were previously covered by $u$ and will be covered by $v$ after the shifting operation (see Fig.~\ref{fig3}).

\begin{figure}[h]
\begin{center}
\includegraphics{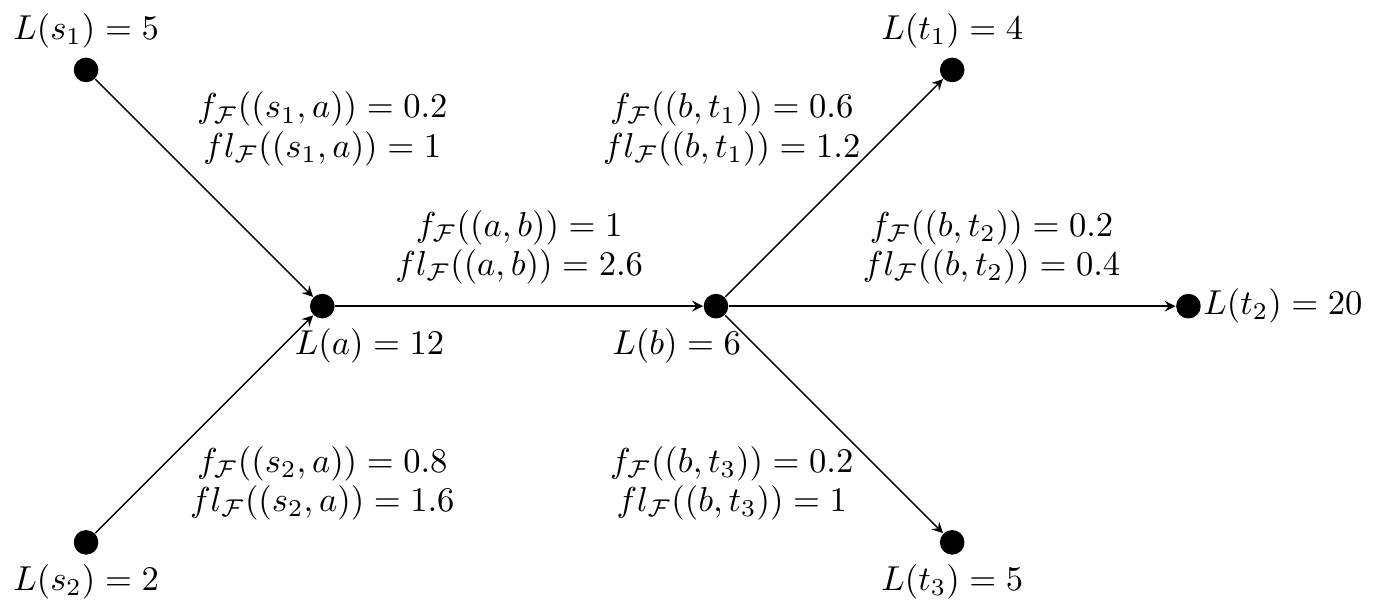}
\caption{The graph $G_{\cF}$ for an acyclic $y$-flow $\cF=\{(0.2,s_1,a,b,t_3),(0.6,s_2,a,b,t_1),(0.2,s_2,a,b,t_2)\}$ from $S=\{s_1,s_2\}$ to $T=\{t_1,t_2,t_3\}$, 
where $y_{s_1}=0.4$, $y_{s_2}=y_a=y_b=1$, $y_{t_1}=0$, $y_{t_2}=0.8$, $y_{t_3}=0.1$.
Note that even though each path in $\cF$ has starting point capacity not greater than its ending point capacity
the vertex $t_1 \in T$ is reachable from $s_1 \in S$ in $G_{\cF}$ despite the fact that $L(s_1) > L(t_1)$.
}
\label{fig3}
\end{center}
\end{figure}

Now we show that if we are given an acyclic $y$-flow $\cF$
then we can transfer $y$-values using a {\em chain shifting} method without increasing the $\radius$ of vertices by too much.
Formal definitions and lemmas follow.


\begin{definition}[\bf chain shifting]
Let $\cF$ be an acyclic $y$-flow from $S$ to $T$ and let $(x,y)$ be a $\delta$-feasible solution.
Let $G_{\cF}=(V,A)$ be the auxiliary acyclic flow graph. 

By {\em chain shifting} we denote the following operation:
\begin{itemize}
  \item For each $u,v \in V$, set $\Delta_{u,v}=0$.
  \item For each arc $(u,a) \in A$ in  reverse topological ordering of $G_\cF$:
    \begin{itemize}
      \item For each $v \in V$, let $\Delta=x_{u,v}fl_{\cF}(u,a)/(L(u)y_u)$,
      set $\Delta_{a,v}=\Delta_{a,v}+\Delta$ and $\Delta_{u,v}=\Delta_{u,v}-\Delta$.
    \end{itemize}
  \item For each $u,v \in V$, set $x_{u,v}=x_{u,v}+\Delta_{u,v}$.
  \item For each $s\in S$ decrease $y_s$ by $\sum_{(s,u) \in A} f_{\cF}((s,u))$.
  \item For each $t\in T$ increase $y_t$ by $\sum_{(u,t) \in A} f_{\cF}((u,t))$.
\end{itemize}
\end{definition}

For a directed graph $G=(V,A)$, for a vertex $v$, we denote $N^{in}(v)=\{u : (u,v) \in A\}$ and $N^{out}(v)=\{u : (v,u) \in A\}$.

\begin{lemma}
\label{lem:chain-shifting}
Let $(x',y')$ be the result of the chain shifting operation on a $\delta$-feasible solution $(x,y)$ according to an acyclic $y$-flow $\cF$ from $S$ to $T$.
If $d$ is the greatest distance in $G$ between two adjacent vertices in $G_{\cF}$, then
$(x',y')$ is a $(\delta+d)$-feasible solution, and for each vertex $v$ of indegree zero in $G_{\cF}$, we have $\radius_{(x',y')}(v) \le \radius_{(x,y)}(v)$,
whereas for other vertices $v$, we have 
$$\radius_{(x',y')}(v) \le \max(\radius_{(x,y)}(v), \max_{a\in N^{in}_{G_{\cF}}(v)} (\radius_{(x,y)}(a)+\dist_G(a,v)))\,.$$
Furthermore for each $v \in V \setminus (S \cup T)$ its $y$-value is the same in $(x,y)$ and $(x',y')$.
\end{lemma}

\begin{proof}
First observe that for any arc $(u,v)$ in $G_{\cF}$ we have the following inequality
\begin{align}
\label{ineq1}
fl_{\cF}(u,v)/L(u) \le f_{\cF}(u,v),
\end{align}
since on each path in $\cF$ the starting vertex has a smallest capacity. Moreover
by the definition of a $y$-flow, through each vertex $u \in V$ at most $y_u$ units are transferred,
hence by Inequality~(\ref{ineq1}) we have:
\begin{align}
\label{ineq2}
\sum_{a\in N^{out}_{G_{\cF}}(u)} fl_{\cF}(u,a)/L(u) \le \sum_{a\in N^{out}_{G_{\cF}}(u)} f_{\cF}(u,a) \le y_u\,\text{and}
\end{align}

\begin{align}
\label{ineq3}
(1-\sum_{a\in N^{out}_{G_\cF}(u)}\frac{fl_{\cF}(u,a)}{L(u)y_u}) \ge 0\,.
\end{align}

We want to show that $(x',y')$ is a $(\delta+d)$-feasible solution.
Observe that by the definition of chain shifting $\sum_{v\in V}y_v' = \sum_{v\in V} y_v$,
and for each vertex $u \in V$ we have $\sum_{v \in V}x_{v,u}' = \sum_{v\in V} x_{v,u}$.
Therefore constraints (\ref{con1}) and (\ref{con4}) are satisfied.
Now we check constraint (\ref{con7}).
Observe that for any $u,v \in V$ we have:
\begin{align*}
x'_{u,v} &=x_{u,v}+\Delta_{u,v} \ge x_{u,v}-\sum_{a\in N^{out}_{G_{\cF}}(u)}x_{u,v}fl_{\cF}(u,a)/(L(u)y_u) \\
 \text{(by (\ref{ineq2}))}  & \ge x_{u_v}-x_{u,v}y_u/y_u = 0\,.
\end{align*}

Next we check constraint (\ref{con2}).
Consider any $u,v \in V$. Observe that if $u \in V \setminus (S \cup T)$ then either $x_{u,v}'=x_{u,v}\le y_u=y_u'$ (which happens if $u$ 
does not belong to any path in $\cF$) or $y_u' = 1 \ge x_{u,v}'$ (when $u$ is an internal vertex in a path from $\cF$).
On the other hand for $s \in S,v\in V$ we have $x_{s,v} \le y_s$ and hence by Inequality~(\ref{ineq1}):

\begin{align*}
x_{s,v}'& =  x_{s,v}(1-\sum_{a \in N^{out}_{G_\cF}(s)}\frac{fl_{\cF}(s,a)}{L(s)y_s}) \le y_s(1-\sum_{a \in N^{out}_{G_\cF}(s)}\frac{f_{\cF}(s,a)}{y_s}) = y_s'\,.
\end{align*}
Similarly for $t \in T, v\in V$ we have
\begin{align*}
x_{t,v}'& =  x_{t,v}+\sum_{a \in N^{in}_{G_\cF}(t)}\frac{fl_{\cF}(a,t)x_{a,v}}{L(a)y_a} \le y_t+\sum_{a \in N^{in}_{G_\cF}(t)}\frac{fl_{\cF}(a,t)}{L(a)} \le y_t+\sum_{a \in N^{in}_{G_\cF}(t)}f_{\cF}(a,t) = y_t'.
\end{align*}

As for constraint (\ref{con3}) for a vertex $u \in V$ we have:
\begin{align*}
\sum_{v \in V} x_{u,v}' &= \sum_{v \in V}x_{u,v}+\sum_{v\in V}\sum_{a \in N^{in}_{G_\cF}(u)}\frac{fl_{\cF}(a,u)x_{a,v}}{L(a)y_a}-\sum_{v\in V}\sum_{a\in N^{out}_{G_\cF}(u)}\frac{fl_{\cF}(u,a)x_{u,v}}{L(u)y_u} \\
 & = \sum_{v \in V}x_{u,v}(1-\sum_{a\in N^{out}_{G_\cF}(u)}\frac{fl_{\cF}(u,a)}{L(u)y_u})+\sum_{v\in V}\sum_{a \in N^{in}_{G_\cF}(u)}\frac{fl_{\cF}(a,u)x_{a,v}}{L(a)y_a} \\
\text{(by (\ref{con3}) of LP1 and (\ref{ineq3}))}  & \le L(u)y_u(1-\sum_{a\in N^{out}_{G_\cF}(u)}\frac{fl_{\cF}(u,a)}{L(u)y_u})+\sum_{a \in N^{in}_{G_\cF}(u)}\sum_{v\in V}\frac{fl_{\cF}(a,u)x_{a,v}}{L(a)y_a} \\
 & = L(u)y_u-\sum_{a\in N^{out}_{G_\cF}(u)}fl_{\cF}(u,a)+\sum_{a \in N^{in}_{G_\cF}(u)}fl_{\cF}(a,u)/(L(a)y_a)\sum_{v\in V}x_{a,v} \\
\text{(by (\ref{con3}) of LP1)}& \le L(u)y_u-\sum_{a\in N^{out}_{G_\cF}(u)}fl_{\cF}(u,a)+\sum_{a \in N^{in}_{G_\cF}(u)}fl_{\cF}(a,u).
\end{align*}
Since for $u \in V\setminus (S \cup T)$ by the definition of $y$-flow $\sum_{a\in N^{out}_{G_\cF}(u)}fl_{\cF}(u,a)=\sum_{a \in N^{in}_{G_\cF}(u)}fl_{\cF}(a,u)$ 
we infer $\sum_{v \in V} x_{u,v}' \le L(u)y_u = L(u)y_u'$.
For $s \in S$ we have $N^{in}_{G_\cF}(s)=\emptyset$ and $fl_{\cF}(s,a)=f_{\cF}(s,a)L(s)$ hence $\sum_{v \in V} x_{s,v}' \le L(s)(y_s-\sum_{a\in N^{out}_{G_\cF}(s)}f_{\cF}(s,a))=L(s)y_s'$.
Similarly for $t \in T$ we have $N^{out}_{G_\cF}(s)=\emptyset$ and $fl_{\cF}(a,t) \ge f_{\cF}(a,t)L(t)$ hence $\sum_{v \in V} x_{s,v}' \le L(t)(y_s+\sum_{a\in N^{in}_{G_\cF}(t)}f_{\cF}(a,t)) = L(t)y_t'$.
Therefore constraint (\ref{con3}) is satisfied.

Constraint (\ref{con5}) may be checked directly by the definition of $y$-flow.
To check the radius of each vertex in $(x',y')$ (that is to verify $(\delta+d)$-feasibility of constraint (\ref{con6}))
we observe that if $x'_{u,v} > 0$ then either $x_{u,v} > 0$ or $x_{a,v} > 0$ for some $a \in N^{in}_{G_\cF}(v)$.

\end{proof}

\subsection{Separable caterpillar structure}
\label{sec:separable}

If we knew that in the caterpillar structure $(P,P')$ produced by Lemma~\ref{lem:path-structure} the capacity of each vertex in $P$
is not smaller than the capacity of each vertex in $P'$ then we could skip this section.
Unfortunately some vertices of $V(P)$ may have smaller capacity than
some vertices of $V(P')$
and for this reason we define the notion of {\em dangerous}, {\em safe} and 
{\em separable} caterpillar structures.

\begin{definition}[\bf safe, dangerous]
\label{def:gamma}
For a caterpillar structure $\cP=(P=(v_1,\ldots,v_p),(v'_0,\ldots,v'_{p+1}))$, by $\Gamma(\cP) \subseteq V(P)$ we denote 
the set containing all vertices $v_i$,
such that there exist $0 \le i_0 < i < i_1 \le p+1$, such that
$v_{i_0}'\not=\nil$, $L(v'_{i_0}) > L(v_i)$ and $v_{i_1}'\not=\nil$, 
$L(v_{i_1}') > L(v_i)$.

A caterpillar structure $\cP$ is {\em safe} if $\Gamma(\cP)=\emptyset$ and {\em dangerous} otherwise.
\end{definition}

\begin{definition}[\bf separable]
\label{def:separable}
Let $(x,y)$ be a $\delta$-feasible solution and let $\cP=(P=(v_1,\ldots,v_p)$, \linebreak $P'=(v_0',\ldots,v_{p+1}'))$ be a dangerous caterpillar structure.
We call $\cP$ {\em separable} iff there exists $1 \le i \le p$ such that $v_i \in \Gamma(\cP)$, $L(v_i)=\min_{v \in \Gamma(\cP)} L(v)$
and either:
\begin{itemize}
  \item $S_1 \ge \lceil S_2 \rceil - S_2$, where $S_2 = \sum_{j=i+1,\ldots,p+1 \atop v_j' \not= \nil} y_{v_j'}$
  and $S_1$ is the sum of values $(1-y_v)$ where $v \in V, v=v_j', L(v) > L(v_i)$ for some $i < j \le p+1$, or,
  \item $S_1 \ge \lceil S_2 \rceil - S_2$, where $S_2 = \sum_{j=0,\ldots,i-1 \atop v_j' \not= \nil} y_{v_j'}$
  and $S_1$ is the sum of values $(1-y_v)$ where $v \in V, v=v_j', L(v) > L(v_i)$ for some $0 \le j < i$.
\end{itemize}
We call such $i$ as above a {\em witness of separability} of $\cP$.
A caterpillar structure that is not separable is called {\em non-separable}.
\end{definition}

The intuition behind the sums $S_1$, $S_2$ is as follows.
The sum $S_2$ contains all the $y$-values of vertices of $P'$ to the right (or left) of $i$.
Since we want to round all the $y$-values of vertices of $P'$, if we want to split the caterpillar structure $(P,P')$
by removing the edge $v_iv_{i+1}$ (or $v_{i-1}v_i$), we need to send $\lceil S_2 \rceil - S_2$
units of flow to the part that does not contain $v_i$, in order to make the sum of $y$-values over all the leaves in both new caterpillar structures integral.
That is to satisfy (\ref{def:path:point8}) of Definition~\ref{def:path-structure}.
In $S_1$ we sum over all vertices, that can potentially receive flow if we start a path at $v_i$,
and the value $(1-y_v)$ is the $y$-value a vertex $v$ may receive.

An example of a separable caterpillar structure is depicted in Fig.~\ref{fig4}. Observe that a non-separable path
structure may be dangerous as in Fig.~\ref{fig5}.

\begin{lemma}
\label{lem:gamma}
Let $\cP=((v_1,\ldots,v_p),(v_0',\ldots,v_{p+1}'))$ be a dangerous caterpillar structure and let $i$ be an index such
that $v_i \in \Gamma(\cP)$ and $L(v_i)=\min_{v \in \Gamma(\cP)} L(v)$.
Moreover let $j$ be an index such that $v_j'\not=\nil,L(v_j') > L(v_i)$.
Then for any $a \in [\min(i,j), \max(i,j)]$ we have $L(v_a) \ge L(v_i)$.
\end{lemma}

\begin{proof}
Assume the contrary, that is assume that such $a$ exists.
Clearly $a \not= i$ (since then $L(v_a)=L(v_i)$)
and also $a \not= j$ because $L(v_j) \ge L(v_j') > L(v_i)$
(where the first inequality follows by (\ref{def:path:point4}) of Def.~\ref{def:path-structure}). 
W.l.o.g. assume that $i < j$ (the other case is symmetric).
By the definition of $\Gamma$ (Def.~\ref{def:gamma}) we infer that there exists $0 \le i_0 < i$
such that $v_{i_0}' \not= \nil$ and $L(v_{i_0}') > L(v_i) > L(v_a)$.
Consequently $v_a \in \Gamma(\cP)$ which contradicts the fact that 
$v_i$ is a smallest capacity vertex in $\Gamma(\cP)$.
\end{proof}

\begin{lemma}
\label{lem:dangerous-non-sep}
Let $\cP=((v_1,\ldots,v_p),(v_0,\ldots,v_{p+1}))$ be a dangerous 
non-separable caterpillar structure and let $\ell=\min_{v\in \Gamma(\cP)} L(v)$.
For $I=\{i\ :\ 0 \le i \le p+1 \wedge v_i'\not=\nil \wedge L(v_i') > \ell \}$ we have $\sum_{i \in I} (1-y_{v_i'}) < 2$.
\end{lemma}

\begin{proof}
Consider any $v_i \in \Gamma(\cP)$ such that $L(v_i) = \ell$.
Let $I_1 = I \cap [0,i-1]$ and $I_2 = I \cap [i+1,p+1]$ (note that $I = I_1 \cup I_2$).
We know that $v_i$ is not a witness of separability
hence each of the two sums $S_1$ in Definition~\ref{def:separable}
is strictly smaller than $1$, since otherwise we would have $S_1 \ge 1 \ge \lceil S_2 \rceil - S_2$.
Consequently $\sum_{i \in I_1} (1-y_{v_i'}) < 1$
and similarly $\sum_{i \in I_2} (1-y_{v_i'}) < 1$.
\end{proof}

\begin{figure}[h]
\begin{center}
\includegraphics{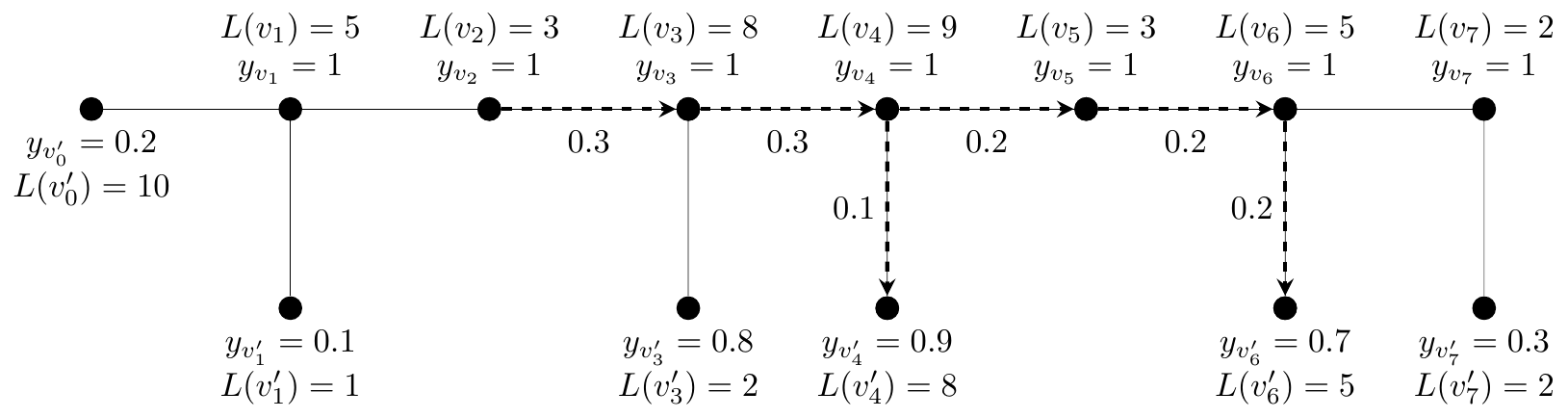}
\caption{A separable caterpillar structure $(P=(v_1,\ldots,v_7),P'=(v_0',v_1',\nil,v_2',v_3',\nil,v_5',v_6',\nil,\nil))$, where $\Gamma((P,P'))=\{v_1,v_2,v_5\}$ (note that $v_7 \not\in \Gamma((P,P'))$, since $v_8'=\nil$).
By dashed edges an acyclic flow $\cF=\{(0.1,v_2,v_3,v_4,v_4'),(0.2,v_2,v_3,v_4,v_5,v_6,v_6')\}$ from $\{v_2\}$ to $\{v_4', v_6'\}$ is marked with values $f_\cF$ printed in the middle of each arc.}
\label{fig4}
\end{center}
\end{figure}

In the following lemma we use the procedure \textsc{Separate}$(\cP=(P,P'))$ which given a $\delta$-caterpillar structure produces a set of non-separable $\delta$-caterpillar structures.
The pseudocode of \textsc{Separate} is given as Algorithm~\ref{alg:separate}
and at high level it performs the following steps:
\begin{enumerate}
  \item(Lines 1-2) If $(P,P')$ is non-separable then return $\cP$.
  \item(Line 3) Otherwise let $i$ be a witness from Definition~\ref{def:separable} with the smallest value of 
  $L(v_i)$ and assume that $i$ is a witness due to the first conditition in the definition (the other case is symmetric).
  \item(Lines 4-5) If the sum $S_2$ is integral, then run the procedure recursively on two caterpillar structures $((v_1,\ldots,v_i),(v_0',\ldots,v_i',\nil))$,
  $((v_{i+1},\ldots,v_p),(\nil,v_{i+1}',\ldots,v_{p+1}'))$ and return the set of obtained non-separable $\delta$-caterpillar structures.
  \item(Lines 6-9)  Let $I=\{j : i < j \le p+1, v_j'\not=\nil, L(v_j') > L(v_i)$, that is the set of indices used
 in the definition of the sum $S_1$ and denote $I=\{i_1,\ldots,i_r\}$, where $i_1 < \ldots < i_r$ and let $\ell_0$ be the smallest integer such 
 that $\sum_{j \in I, j \le \ell_0} (1-y_{v_j'}) \ge \lceil S_2 \rceil - S_2$ (such $\ell_0$ exists since $i$ is a witness of separability).
 \item(Lines 10-11) Construct an acyclic $y$-flow $\cF$ from $\{v_i\}$ to $\{v_j': j \in I, j \le \ell_0\}$.
   That is for each $j \in I$, $j < \ell_0$ add to $\cF$ a sequence $(1-y_{v_j'},v_i,v_{i+1},\ldots,v_j,v_j')$
   and for $j=\ell_0$ add to $\cF$ a sequence $(\alpha,v_i,v_{i+1},\ldots,v_{\ell_0},v_{\ell_0}')$,
   where $\alpha=(\lceil S_2 \rceil - S_2) - \sum_{j \in I, j < \ell_0}(1-y_{v_j'})$ (see Fig.~\ref{fig4} for illustration);
   perform chain shifting according to $\cF$.
  \item(Lines 12-18) If $i < p$ then run the procedure recursively on the caterpillar structure $((v_{i+1},\ldots,v_p)$, $(\nil,v_{i+1}'',\ldots,v_{p+1}''))$,
 where $v_{j}''=v_j'$ if $y_{v_j'} < 1$ and $v_j''=\nil$ otherwise for $i+1\le j \le p+1$ (if $i=p$ we already know that $y_{v_{p+1}'}=1$).
  \item(Lines 19-21) If $i = 1$ then perform group shifting on the set $\{v_0',v_1,v_1'\} \setminus \{\nil\}$ and return.
  \item(Lines 22-23) If $v_i'\not=\nil$ then shift $\min(y_{v_i'}, 1-y_{v_i})$ from $y_{v_i'}$ to $y_{v_i}$ (by the properties
    of a caterpillar structure we know that $L(v_i) \ge L(v_i')$).
  \item(Lines 24-29) If $y_{v_i}=1$ then run the procedure recursively on the caterpillar structure $((v_1,\ldots,v_i)$, $(v_0',\ldots,v_{i-1}',u,\nil))$,
   where $u=v_i'$ if $v_i' \not= \nil, y_{v_i'}>0$, or $u=\nil$ otherwise.
  \item(Lines 30-31) Otherwise (if $y_{v_i} < 1$) run the procedure recursively on the caterpillar structure $((v_1,\ldots,v_{i-1})$, $(v_0',\ldots,v_{i-1}',v_i))$.
\end{enumerate}

\begin{algorithm}[thp]
\caption{\algo{\textsc{Separate}}\label{alg:separate}}
\textbf{Input:} A $\delta$-caterpillar structure $(\cP=(P,P'))$.\\
\textbf{Output:} A set of non-separable $\delta$-caterpillar structures $\cS$, such that each vertex $v$ which belongs to $\cP$ and does not belong to any caterpillar structure of $\cS$, has integral $y$-value.

\begin{algorithmic}[1]

\IF {$(P,P')$ is non-separable}
\RETURN $\{\cP\}$
\ENDIF
\STATE $i\gets$ a witness from Definition~\ref{def:separable} with the smallest value of $L(v_i)$
\begin{flushright}
 {\it /*assume that $i$ is a witness due to the first condition*/}\\
 {\it /*in the definition, since the other case is symmetric*/}
\end{flushright}
\IF {$S_2$ is integral}
  \RETURN \textsc{Separate}$((v_1,\ldots,v_i),(v_0',\ldots,v_i',\nil))\ \cup\ $\textsc{Separate}$((v_{i+1},\ldots,v_p),(\nil,v_{i+1}',\ldots,v_{p+1}'))$ 
\ENDIF
\STATE $\cS\gets\emptyset$
\STATE $I\gets \{j : i < j \le p+1, v_j'\not=\nil, L(v_j') > L(v_i)\}$
  \begin{flushright}
{\it /*$I$ is the set of indices used in the definition of the sum $S_1$*/}
  \end{flushright}
\STATE Denote $I=\{i_1,\ldots,i_r\}$, where $i_1 < \ldots < i_r$
\STATE $\ell_0\gets$ the smallest integer such that $\sum_{j \in I, j \le \ell_0} (1-y_{v_j'}) \ge \lceil S_2 \rceil - S_2$ 
\begin{flushright}
 {\it /*such $\ell_0$ exists since $i$ is a witness of separability*/}
\end{flushright}
\STATE Construct an acyclic $y$-flow $\cF$ from $\{v_i\}$ to $\{v_j': j \in I, j \le \ell_0\}$.  That is for each $j \in I$, $j < \ell_0$ add to $\cF$ a sequence $(1-y_{v_j'},v_i,v_{i+1},\ldots,v_j,v_j')$ and for $j=\ell_0$ add to $\cF$ a sequence $(\alpha,v_i,v_{i+1},\ldots,v_{\ell_0},v_{\ell_0}')$, where $\alpha=(\lceil S_2 \rceil - S_2) - \sum_{j \in I, j < \ell_0}(1-y_{v_j'})$ (see Fig.~\ref{fig4} for illustration).
\STATE Perform chain shifting according to $\cF$.
\IF {$i < p$}

\item[]
\begin{flushright}
{\it /*if $i=p$ we already know that $y_{v_{p+1}'}=1$*/}
\end{flushright}
  \FOR {j=i+1 {\bf to} p+1}
    \IF {$v_j'\not=\nil$ {\bf and} $y_{v_j'} < 1$} 
      \STATE $v_{j}'' \gets v_j'$
    \ELSE
      \STATE $v_j''\gets \nil$
    \ENDIF
  \ENDFOR
  \STATE $\cS=\cS\ \cup\ $\textsc{Separate}$((v_{i+1},\ldots,v_p),(\nil,v_{i+1}'',\ldots,v_{p+1}''))$,
\ENDIF

\IF {$i = 1$}
  \STATE Perform group shifting on $\{v_0',v_1,v_1'\} \setminus \{\nil\}$
  \RETURN $\cS$
\ENDIF

\IF {$v_i'\not=\nil$}
  \STATE Shift $\min(y_{v_i'}, 1-y_{v_i})$ from $y_{v_i'}$ to $y_{v_i}$ 
  \begin{flushright}
   {\it /*by the properties of a caterpillar structure we know that $L(v_i) \ge L(v_i')$*/}
  \end{flushright}
\ENDIF

\IF {$y_{v_i}=1$}
  \IF {$v_i' \not= \nil$ {\bf and} $y_{v_i'}>0$}
    \STATE $u=v_i'$
  \ELSE
     \STATE $u=\nil$
  \ENDIF
  \STATE $\cS\gets \cS\ \cup\ $\textsc{Separate}$((v_1,\ldots,v_i),(v_0',\ldots,v_{i-1}',u,\nil))$,
\ELSE
  \STATE $\cS\gets \cS\ \cup\ $\textsc{Separate}$((v_1,\ldots,v_{i-1}),(v_0',\ldots,v_{i-1}',v_i))$
\ENDIF
\RETURN $\cS$
\end{algorithmic}
\end{algorithm}

\begin{lemma}
\label{lem:non-sep-path-structures}
For a given feasible LP solution $(x,y)$ we can find a $68$-feasible solution $(x',y')$
together with a set of vertex disjoint non-separable $21$-caterpillar structures $\mathcal{S}$ such
that each vertex $v$ outside of the set has an integral $y$-value in $(x',y')$.
Furthermore for each vertex $v$ that belongs to some caterpillar structure from $\mathcal{S}$
we have $\radius_{(x',y')}(v) \le 47$.
\end{lemma}

\begin{proof}
First we use Lemma~\ref{lem:path-structure} to obtain a $5$-feasible solution
$(x',y')$ together with a $21$-caterpillar structure $(P_0,P_0')$ such that 
each vertex outside $(P_0,P_0')$ has an integral $y$-value in $(x',y')$.
Next we run the procedure \textsc{Separate} on the caterpillar structure $(P_0,P_0')$ using $(x',y')$ as the assignment.

Let us prove, that given a $\delta$-caterpillar structure the procedure \textsc{Separate}
indeed returns a set of non-separable $\delta$-caterpillar structures.
By Lemma~\ref{lem:gamma} the set $\cF$ created in Line~10 of the procedure \textsc{Separate}
is indeed an acyclic $y$-flow.
Furthermore caterpillar structures created in Lines 5, 18, 29, 31.
are indeed $\delta$-caterpillar structures since they satisfy all the conditions of Definition~\ref{def:path-structure}.
Moreover observe that if a vertex $v$ belongs to $\Gamma(\cP')$, where $\cP'$ is a caterpillar structure created in
one of Lines 5, 18, 29, 31,
then $v \in \Gamma(\cP)$.

Since some vertices of the caterpillar structure created in Lines~18, 29 and 31
may have increased radius we need to argue why the radius does not grow too much in the subsequent recursive calls.

Let $\mathcal{S}$ be the set of non-separable $21$-caterpillar structures, which
is a result of \textsc{Separate}$(P_0,P_0')$.
Observe that the caterpillar structures are vertex disjoint
and moreover each vertex $v$ which does not belong to any caterpillar structure from $S$
has an integral $y$-value in $(x',y')$.

Due to Lemmas~\ref{lem:shift},\ref{lem:group-shift},\ref{lem:chain-shifting} when we modify the assignment $(x',y')$
by shifting, group shifting or chain shifting we satisfy all constraints of the LP possibly except~(\ref{con6}).
To show $68$-feasibility of $(x',y')$ we prove that for each caterpillar structure $((v_1,\ldots,v_p),(v_0',\ldots,v_{p+1}'))$ used as an argument of the procedure \textsc{Separate}
there exist two indices $0 \le i_0 \le j_0 \le p+1$ such that:
\begin{itemize}
  \item for each $i \in [1,p]$ we have $\radius_{(x',y')}(v_i) \le 5+21\cdot[i \le i_0]+21\cdot[i \ge j_0]$,
  \item for each $i \in [0,p+1]$ if $v_i' \not= \nil$ then $\radius_{(x',y')}(v_i') \le 5+21\cdot[i \le i_0]+21\cdot[i \ge j_0]$,
  \item for each $i \in [0,i_0) \cup (j_0,p+1]$ such that $v_i' \not= \nil$ for any $v \in \Gamma(\cP)$ we have $L(v_i') \le L(v)$.
\end{itemize}
In the formulas above we use the Iverson's bracket notation, that is $[\phi]$ equals $1$ when $\phi$ is true, and $0$ otherwise.
We call a caterpillar structure satisfying the above properties a {\em good} caterpillar structure.
The initial caterpillar structure $(P_0,P_0')$ is good since we can set $i_0=0$ and $j_0=p+1$.
Hence we assume that a caterpillar structure $(P,P')$ given as an argument to the \textsc{Separate} procedure is good
and we want to show that all recursive calls are given good caterpillar structures.
If the procedure exits in Line~2 then it is non-separable,
thus we assume that $(P,P')$ is separable and due to the symmetry w.l.o.g. we may assume that 
$i$ is a witness satisfying the first condition of Definition~\ref{def:separable}.
Clearly caterpillar structures constructed in Line~5 are good caterpillar structures.
Now we investigate caterpillar structures constructed in Lines~18,~29 and 31.
Observe that due to the definition of a good caterpillar structure we can prove that $i_0 < i < j_0$.
Indeed if $i \in [1,i_0]$ then there is no $0 \le i' < i$ with $v_i' \not=\nil$ and $L(v_i') > L(v_i)$
and similarly for $i \in [j_0,p]$ there is no $i < i' \le p+1$ with $v_i' \not=\nil$ and $L(v_i') > L(v_i)$.
Note that for $\ell_0$ defined in Line~9 we have $\ell_0 \le j_0$ since due to the last property of good caterpillar structure
$I \subseteq [i+1,j_0]$.
Furthermore for each $j \in [i+1,\ell_0)$ either $v_j'=\nil$ or $L(v_j') \le L(v_i)$ or 
after chain shifting by Lemma~\ref{lem:chain-shifting} the vertex $v_j'$ will have integral $y$-value equal to one
and will be not be included in caterpillar structures handled recursively.
Consequently the caterpillar structure constructed in Line~18 is a good caterpillar structure.
Finally consider caterpillar structures created in Lines~29 and 31
and observe that a pair of indices $i_0,i$ satisfies the properties of a good caterpillar structure.

Since all the caterpillar structures from the set $\mathcal{S}$ are good caterpillar structures
the last part to show before we claim that $(x',y')$ is a $68$-feasible solution 
is that vertices from $X=(V(P_0)\cup V(P_0')) \setminus (\bigcup_{(P,P') \in \mathcal{S}} V(P)\cup V(P'))$
have radius bounded by $68$ (by the definition of a good caterpillar structure all vertices from $\bigcup_{(P,P') \in \mathcal{S}} V(P)\cup V(P')$
have radius bounded by $47$).
The only possibilities for a vertex $v$ to become a member of $X$ is when its $y$-value becomes integral,
which may happen in Line~17 (when $y_{v_j'}'=1$), Line~20 (for each non-nil vertex in $\{v_0',v_1,v_1'\}$), 
Line~23 (when $y_{v_i'}'=0$).
However in all cases because $i_0 < i < j_0$ and by Lemma~\ref{lem:chain-shifting} we infer that $\radius_{(x',y')}(v) \le 68$.
\end{proof}

\begin{figure}[h]
\begin{center}
\includegraphics{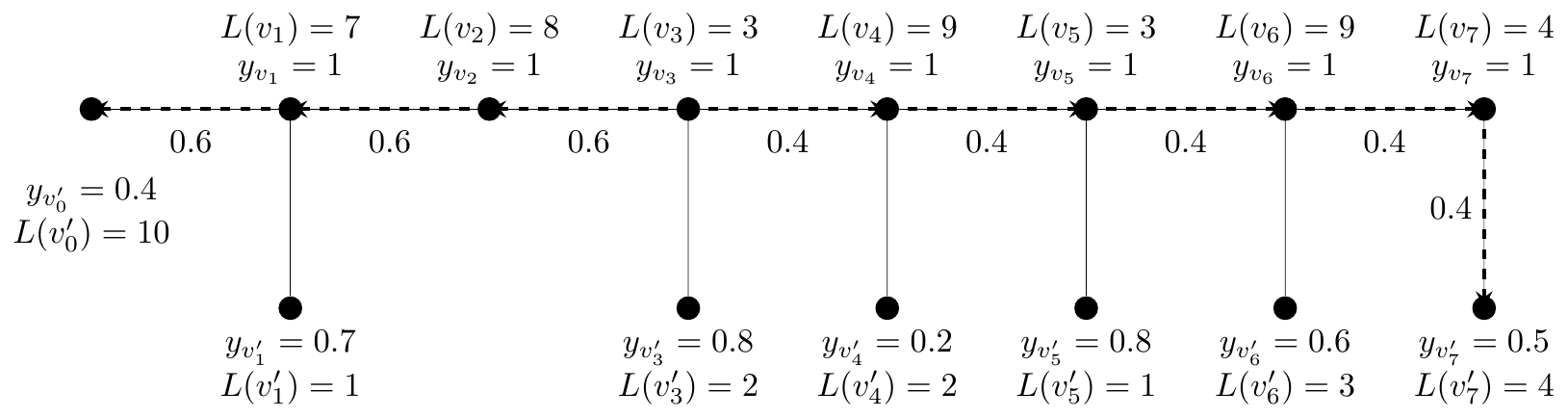}
\caption{A dangerous caterpillar structure $(P=(v_1,\ldots,v_7),P'=(v_0',v_1',\nil,v_2',v_3',v_4',v_5',v_6',v_7',\nil))$, where $\Gamma((P,P'))=\{v_3,v_5\}$.
The caterpillar structure is non-separable because both for $i=3$ and $i=5$ in Definition~\ref{def:separable} the sum $S_1$ is at most $0.6$, while $\lceil S_2 \rceil - S_2$ is equal to $0.9$.
By dashed edges an acyclic flow $\cF=\{(0.6,v_3,v_2,v_1,v_0'),(0.4,v_3,v_4,v_5,v_6,v_7,v_7')\}$ from $\{v_3\}$ to $\{v_0', v_7'\}$ is marked with values $f_\cF$ printed in the middle of each arc.}
\label{fig5}
\end{center}
\end{figure}

In the following lemma we transform non-separable caterpillar structures into safe caterpillar structures.

\begin{lemma}
\label{lem:safe-path-structures}
There exist constants $c,\delta$ such that 
for a given feasible LP solution $(x,y)$ we can find a $c$-feasible solution $(x',y')$
together with a set of vertex disjoint safe $\delta$-caterpillar structures $\mathcal{S}$ such
that each vertex $v$ outside of the set has an integral $y$-value in $(x',y')$.
\end{lemma}

\begin{proof}
We use Lemma~\ref{lem:non-sep-path-structures} to obtain a set $\mathcal{S}$
of vertex disjoint non-separable $21$-caterpillar structures.
Our goal is to transform each dangerous caterpillar structure in $\mathcal{S}$
into a safe caterpillar structure.

Consider a dangerous non-separable $\delta_0$-caterpillar structure $\cP=((v_1,\ldots,v_p),(v_0',\ldots,v_{p+1}')) \in \mathcal{S}$ and let $v_a$ be a minimum capacity vertex in $\Gamma(\cP)$.
Moreover let $I=\{i\ :\ 0 \le i \le p+1 \wedge v_i'\not=\nil \wedge L(v_i') > L(v_a) \}$.
Construct any acyclic $y$-flow which sends $\min(1,\sum_{i \in I} (1-y_{v_i'}))$
from $\{v_a\}$ to $\{v_i'\ :\ i \in I\}$ (see Fig.~\ref{fig5}).
Such flow always exists due to Lemma~\ref{lem:gamma}.

Let $Y=\{v_a,v_a',v_{a-1}'\} \setminus \{\nil\}$ and perform group shifting
on $Y$ (note that $a \ge 1$, since $v_a \in \Gamma(\cP)$).
Replace $\cP$ in $\mathcal{S}$ with the $(2\delta_0)$-caterpillar
structure $((v_1,\ldots,v_{a-1},v_{a+1},\ldots,v_p)$, \linebreak $(v_0',\ldots,v_{a-2}',u,v_{a+1}',\ldots,v_{p+1}'))$, where as $u$ we set the only vertex from $Y$ with fractional
$y$-value after group shifting or we set $u=\nil$ if all vertices in $Y$
have integral $y$-values.
We need to argue, that when $u \not=\nil$, we have $L(u) \le L(v_{a-1})$, in order to satisfy (\ref{def:path:point4}) of Definition~\ref{def:path-structure}.
Observe, that if $v_a' \not= \nil$, then $L(v_a') \le L(v_a)$,
and similarly if $v_{a-1}' \not= \nil$, then $L(v_{a-1}') \le L(v_{a-1})$.
Hence to show $L(u) \le L(v_{a-1})$ it is enough to show $L(v_a) \le L(v_{a-1})$,
but this follows from Lemma~\ref{lem:gamma},
since $v_a \in \Gamma(\cP)$.

Note, that each caterpillar structure
will be modified according to the above procedure at most twice,
since after one iteration the sum $\sum_{i \in I} (1-y_{v_i'})$
either equals zero or decreases by one, and by Lemma~\ref{lem:dangerous-non-sep} we have $\sum_{i \in I} (1-y_{v_i'}) < 2$.
Consequently by Lemmas~\ref{lem:chain-shifting}, \ref{lem:group-shift}
we obtain the desired set of vertex disjoint $\delta$-caterpillar structure
together with a $c$-feasible solution.
\end{proof}

\subsection{Rounding safe caterpillar structures}
\label{sec:rounding-flow}
In this section we describe how to round the $c$-feasible solution $(x',y')$ 
using the set of vertex disjoint safe caterpillar structures $\mathcal{S}$ from Lemma~\ref{lem:safe-path-structures}.
In order to do that we introduce a notion of {\em rounding flow} which is a special kind of $y$-flow
defined for a caterpillar structure.

\begin{definition}[\bf rounding flow]
\label{def:rounding-flow}
For a caterpillar structure $(P,P')$ and an assignment $(x,y)$ we call
$\cF$ a {\em rounding flow} iff $\cF$ is a $y$-flow from $S$ to $T$ where $S \cup T= V(P')$,
for each $v_i' \in S$ we have $f_{\cF}((v_i',v_i))=y_{v_i'}$
and for each $v_i' \in T$ we have $f_{\cF}((v_i,v_i'))=1-y_{v_i'}$.
Furthermore each flow path from $\cF$ can not go through a vertex from $V \setminus (V(P) \cup V(P'))$.
\end{definition}

In order to obtain a rounding flow for each vertex of $V(P')$
(which by definition have fractional $y$-values), we have to decide whether it will be a source (member of $S$)
or a sink (member of $T$).
After chain shifting according to $\cF$ all sources should have $y$-value equal to zero whereas
all sinks should have $y$-value equal to one
and consequently all vertices from the caterpillar structure will have integral $y$-value.
In the following lemma we show that for each non-separable caterpillar structure we can always find
a rounding flow in polynomial time.

\begin{lemma}
\label{lem:rounding-flow}
For any safe $\delta$-caterpillar structure $(P,P')$ and an assignment $(x,y)$ 
there exists a rounding flow $\cF$ such that for any two adjacent
vertices in $G_{\cF}$ their distance in $G$ is at most $\delta$.
Furthermore we can find such a rounding flow in polynomial time.
\end{lemma}

\begin{proof}
We present a recursive procedure which constructs a desired rounding flow.
Note that some recursive calls of the procedure might potentially involve 
infeasible assignments $(x',y')$, however we prove that if the initial
call gives the procedure a safe $\delta$-caterpillar structure, then as a result we obtain a valid rounding flow.

Let us describe a procedure which is given a caterpillar structure $(P,P')$ together with an assignment $y$ (the procedure does not need 
the $x$ part of an assignment).
Denote $P=(v_1,\ldots,v_p)$ and $P'=(v_0',\ldots,v_{p+1}')$.
If $V(P')=\emptyset$ then we simply return the empty rounding flow.
Otherwise let $i$ be the smallest integer such that the sum of $y$-values of $X=\{v_0',\ldots,v_i'\} \setminus \nil$ 
is at least one (such $i$ always exists since the sum of all $y$-values in $V(P')$ is integral by (\ref{def:path:point8}) of Def.~\ref{def:path-structure}).
Note that since all vertices in $V(P')$ have fractional $y$-values we have $i>0$.
Let $0 \le i_0 \le i$ be an index such that $v_{i_0}'\not=\nil$ and $v_{i_0}'$ has the biggest capacity in $X$.
Let $\alpha=\sum_{v \in X} y_v$.
If $\alpha=1$ then we recursively construct a rounding flow $\cF$ from $S$ to $T$ for a smaller caterpillar structure $((v_{i+1},\ldots,v_{p}),(\nil,v_{i+1}',\ldots,v_{p+1}'))$
and (i) add to $S$ the set of vertices $X \setminus \{v_{i_0}'\}$ (ii) add to $T$ the vertex $v_{i_0}'$ (iii) 
for each $v_j'\in X\setminus \{v_{i_0}'\}$ add to $\cF$ 
a flow path $(y_{v_j'}, v_j', v_j, \ldots, v_{i_0}, v_{i_0}')$.
In this case we return $\cF$ as a desired rounding flow for $(P,P')$.
Hence from now on we assume $\alpha > 1$ and $\alpha - 1 < y_{v_{i}'}$.
Consider two cases: $i_0 < i$ and $i_0 = i$.

First let us assume that $i_0 < i$.
We store $z:=y_{v_i'}$ and temporarily set $y_{v_i'}=\alpha-1$.
Next recursively construct a rounding flow $\cF$ from $S \subseteq V(P'')$ to $T \subseteq V(P'')$ for a smaller caterpillar structure $((v_i,\ldots,v_p),P'')$, where
$P''=(\nil,v_i',\ldots,v_{p+1}')$ (note that the sum of $y$-values in $P''$ is integral).
Now consider two cases:
\begin{itemize}
  \item if $v_i' \in S$ then: (i) add to $S$ vertices from $X\setminus \{v_i', v_{i_0}'\}$ (ii) add to $T$ the vertex $v_{i_0}'$ (iii) 
  for each $v_j'\in X\setminus \{v_{i_0}', v_i'\}$ add to $\cF$ 
  a flow path $(y_{v_j'}, v_j', v_j, \ldots, v_{i_0}, v_{i_0}')$ (iv) add to $\cF$ a flow path $(z-y_{v_i'}, v_i', v_i, \ldots, v_{i_0}, v_{i_0}')$ (v) set
   $y_{v_i'}:=z$ (vi) return $\cF$.
  \item if $v_i' \in T$ then: (i) add to $S$ vertices from $X\setminus \{v_i', v_{i_0}'\}$ (ii) add to $T$ the vertex $v_{i_0}'$ 
  (iii) out of the flow paths in $\cF$ that end in $v_i'$ leave only that many, that send exactly $1-z$ units of flow 
  and reroute the rest paths to $v_{i_0}'$ through vertices $v_{i-1},v_{i-2},\ldots,v_{i_0}$ 
  (iv) for each $v_j'\in X\setminus \{v_{i_0}',v_i'\}$ add to $\cF$ 
  a flow path $(y_{v_j'}, v_j', v_j, \ldots, v_{i_0}, v_{i_0}')$ (v) return $\cF$.
\end{itemize}

Now assume that $i_0=i$. We create a smaller caterpillar structure $((v_a,v_{i+1},v_{i+2},\ldots,v_{p})$, \linebreak $(\nil, v_a', v_{i+1}', \ldots, v_{p+1}'))$,
where $v_a,v_a'$ are two newly created vertices with $y_{v_a'}:=\alpha-1$ and $L(v_a'):=L(v_a):=L(v_{i_1}')$,
where $v_{i_1}'$ is the second biggest capacity vertex in the set $X$.
Next run recursively our procedure on the newly created caterpillar structure to obtain a rounding flow $\cF$ from $S$ to $T$.
Again, consider two cases:
\begin{itemize}
  \item if $v_a' \in S$ then: (i) set $S := (S \setminus \{v_a'\}) \cup (X \setminus \{v_i'\})$ (ii) set $T := T \cup \{v_i'\}$
  (iii) change in $\cF$ all the paths that start in $v_a'$ to start in $X \setminus \{v_i'\}$
  (iv) add to $\cF$ paths that start in $X$ and transfer $1-y_{v_i'}$ units of flow from $X$ to $v_i'$
  (v) return $\cF$.
  \item if $v_a' \in T$ then: (i) set $S := S \cup (X \setminus \{v_i', v_{i_1}'\})$ (ii) set $T := (T \setminus \{v_a'\}) \cup \{v_i', v_{i_1}'\}$
  (iii) reroute some of the flow paths from $\cF$ that end in $v_a'$ to that transfer exactly $1-y_{v_i'}$
  units of flow to $v_{i'}$ (that is remove $v_a'$ as the last vertex on those paths and 
      extend the paths by $v_i,v_i'$) (iv) reroute all the remaining flow paths in $\cF$ that end in $v_a'$
  to $v_{i_1}'$ (that is remove $v_a'$ and extend those paths by $v_i,v_{i-1},\ldots,v_{i_1},v_{i_1}'$)
  (v) for each $v_j'\in X\setminus \{v_i',v_{i_1}'\}$ add to $\cF$ 
  a flow path $(y_{v_j'}, v_j', v_j, \ldots, v_{i_1}, v_{i_1}')$ (v) return $\cF$.
\end{itemize}

Finally we prove that if the procedure receives a safe caterpillar structure then
it returns a desired rounding flow.
The only property of the rounding flow that needs detailed analysis is the
assumption that each internal vertex of a flow path has capacity not smaller
than its the capacity of its starting point.
Let us assume that there exists a path in $\cF$
that starts in $v_a'$, goes though $v_b$ and ends in $v_c'$,
where $L(v_c') \ge L(v_a') > L(v_b)$.
This contradicts the assumption that $\cP$ is safe because $v_b \in \Gamma(\cP)$.
\end{proof}

The following theorem summarizes Sections~\ref{sec:path-structure}, \ref{sec:chain-shifting}, \ref{sec:separable}, \ref{sec:rounding-flow}.

\begin{theorem}
\label{thm:integral-y}
For a connected graph $G$, if LP1 has a feasible solution then
we can find a $c$-feasible solution with integral $y$-values.
\end{theorem}

\begin{proof}
Using a feasible solution to LP1,
by Lemma~\ref{lem:safe-path-structures},
we obtain a $c$-feasible solution $(x',y')$,
together with a set of vertex disjoint safe $\delta$-caterpillar structures $\mathcal{S}$, such that vertices that do not belong to any caterpillar structure in $\mathcal{S}$
have integral $y$-value in $(x',y')$.
Next by Lemma~\ref{lem:rounding-flow} for each $\delta$-caterpillar structure $(P,P') \in \mathcal{S}$ we find a rounding flow $\cF_{(P,P')}$.
Finally for each $\delta$-caterpillar structure $(P,P')$ we perform 
chain shifting with respect to $\cF_{(P,P')}$, and by Lemma~\ref{lem:chain-shifting}
we obtain a $c'$-feasible solution $(x'',y'')$ to LP1.

By Lemma~\ref{lem:safe-path-structures}, vertices outside of $\mathcal{S}$
have integral $y$-value in $(x',y')$. Moreover by Definition~\ref{def:rounding-flow},
after chain shifting all the vertices in each caterpillar structure of $\mathcal{S}$ have integral $y$-values in $(x'',y'')$.
\end{proof}

\subsection{Rounding $x$-values}
\label{sec:rounding-x}

In this section we show how to extend Theorem~\ref{thm:integral-y} to obtain not only integral $y$-values, but also integral $x$-values.
The following lemma is standard (using network flows).

\begin{lemma}
Let $(x,y)$ be a $\delta$-feasible solution such that all $y$-values are integral.
There is a polynomial time algorithm that creates a $\delta$-feasible solution
which has both $x$- and $y$-values integral.
\end{lemma}

As a consequence of Theorem~\ref{thm:integral-y} and the above lemma the proof of Theorem~\ref{thm:main} follows.

\section{Soft capacities}
\label{sec:soft}
\newcommand{\softconst}{11}

This section is devoted to the soft capacities variant,
where one can open an arbitrary number of centers in a node.
We present an $\softconst$-approximation algorithm
for the soft capacitated version of the $k$-center problem
with non-uniform capacities, i.e. we prove Theorem~\ref{thm:soft-apx}.
Let us recall that for this problem the LP relaxation for the natural IP,
which we denote as LP1 has the following form.

\begin{align}
\label{scon1}    & \textstyle{\sum_{u \in V}y_{u} = k;} &  & \\
\label{scon2}    & \textstyle{x_{u,v} \le y_u} & \textstyle{\forall u,v \in V} & \\
\label{scon3}    & \textstyle{\sum_{v \in V} x_{u,v} \le L(u)y_u} & \textstyle{\forall u \in V} & \\
\label{scon4}    & \textstyle{\sum_{u \in V} x_{u,v} = 1} & \textstyle{\forall v \in V} & \\
\label{scon6}    & \textstyle{x_{u,v} = 0} & \textstyle{\forall u,v \in V\ \dist_G(u,v) > 1} &  \\
\label{scon7}    & \textstyle{x_{u,v} \ge 0} & \textstyle{\forall u,v \in V} & 
\end{align}

The following lemma is proved by Khuller and Sussman.

\begin{lemma}[\cite{KS}]
\label{lem:ks}
For a connected graph $G=(V,E)$ one can in polynomial time construct
an inclusionwise maximal independent set $S \subseteq V$ in $G^2$,
such that $G^3[S]$ is connected.
\end{lemma}

The proof of the main theorem of this section is inspired by
the algorithm of~\cite{KS} for uniform soft capacities.

\begin{theorem}
There is a polynomial time algorithm, which given an instance of the soft-capacitated $k$-center problem for a connected graph, and a fractional feasible solution for LP1,
can round it to an integral solution that uses non-zero $x_{u,v}$ variables 
for pairs of nodes with distance at most $\softconst$.
\end{theorem}

\begin{proof}
First, we construct the set $S$ using the algorithm from Lemma~\ref{lem:ks}.
Observe, that by constraints (\ref{scon4}) and (\ref{scon2})
for each $v \in V$ we have 
\begin{equation*}
\sum_{u \in N[v]} y_u \ge \sum_{u \in N[v]} x_{u,v} = 1\,.
\end{equation*}
Since $S$ is independent in $G^2$, by constraint (\ref{scon1})
we infer that $|S| \le k$.

We prove the theorem in two steps. First,
we construct a function $\phi : V \rightarrow S$,
such that for each $v \in V$ we have $\dist_G(v,\phi(v)) \le 5$
and for each $s \in S$ we have $|\phi^{-1}(s)| \le \max_{v \in V, \dist_G(v,s) \le 6} L(v)$,
that is we assign each vertex to an element of $S$ within distance $5$,
but we increase the capacity of each $s$ to the maximum capacity reachable within distance $6$.
In the second step we reassign vertices in such a way, that the maximum distance
is at most $11$, and at the same time capacity constraints are satisfied.

Let $T$ be any spanning tree of $G^3[S]$, rooted 
at an arbitrary vertex $r \in S$.
Consider the following procedure, which constructs a function $u : S \rightarrow \mathbb{R}_+\cup \{0\}$ assigning
a tentative fractional number of centers to open in each of the vertices in $S$.
\begin{enumerate}
  \item Initially set $u(s)=0$ for each $s \in S$.
  \item \label{step2} For each vertex $v \in V$, if $\dist_G(v,S) \le 1$, then by the fact that $S$ is independent in $G^2$ 
  there exists exactly one vertex $s_v \in S$, such that $\dist_G(v,s_v) \le 1$, otherwise (when $\dist_G(v,S) = 2$),
  as $s_v$ set any vertex from $S$ within distance $2$ from $v$. 
  \item \label{step3} For each vertex $v \in V$ increase $u(s_v)$ by $y_v$ (note that after this operation for each $s \in S$ we have $u(s) \ge 1$ and $\sum_{s \in S} u(s) = k$).
  \item \label{step4} For each vertex $s \in S$ in a bottom-up order with respect to $T$, let $x=u(s) - \lceil u(s) \rceil$, decrease $u(s)$ by $x$
  and increase $u(p(s))$ by $x$, where $p(s)$ is the parent of $s$ in the tree $T$ (assume that $p(r)=r$).
\end{enumerate}
The last step of the above process ensures, that the function $u$ has only integral positive values and moreover $\sum_{s\in S} u(s) = k$.
We claim that if we open exactly $u(s)$ centers in a vertex $s$, for each $s \in S$, then there exists 
an assignment $\phi: V \rightarrow S$,
such that for each $v \in V$ we have $\dist_G(v,\phi(v)) \le 5$
and for each $s \in S$ we have $|\phi^{-1}(s)| \le \max_{v \in V, \dist_G(v,s) \le 6} L(v)$.
Note that if we know that it exists, then we can find it in polynomial time by using maximum flow computation.
To show the existence of $\phi$, observe that it is enough to show a function $f:V \times S \rightarrow \mathbb{R}_+\cup \{0\}$, satisfying:
\begin{itemize}
  \item for each $v \in V$ we have $\sum_{s \in S} f(v,s) = y_v$,
  \item for each $s\in V$ we have $\sum_{v \in V} f(v,s) = u(s)$,
  \item for each $v \in V,s \in S$ if $f(v,s) > 0$, then $\dist_G(v,s) \le 5$.
\end{itemize}
Less formally, the function $f(v,*) : S \rightarrow \mathbb{R}_+ \cup \{0\}$ is a distribution of the value $y_v$ among
vertices of $S$ within distance $5$. 
Such a function guarantees that we can fractionally cover all the vertices within distance $5$ when we open
an integral number of centers $u(s)$ in each vertex $s \in S$, therefore we can also cover vertices of $V$ integrally, which proves
the existence of the desired assignment $\phi$.
Observe, that we can construct the function $f$ while performing the bottom-up process in Step~\ref{step4}, 
where for each vertex $v \in V$ we split the value $y_v$ between $f(v,s_v)$ and $f(v,p(s_v))$. 
It is always possible, because while going up the tree we can ensure that the part of $y_v$, that was assigned to $f(v,p(s_v))$ remains
in $p(s_v)$, since we send up $u(s) - \lceil u(s) \rceil < 1$ and after Step~\ref{step3} we have $u(s') \ge 1$ for each $s' \in S$.
Consequently, we know that there exists the desired assignment $\phi$.

Finally, we construct the final assignment $\phi_0 : V \rightarrow S'$ as follows. 
For each $s \in S$ let $$s_L = \arg \max_{v \in V, \dist_G(s,v) \le 6}L(v)\,,$$ open $u(s)$ centers in $s_L$
and assign all the vertices of $\phi^{-1}(s)$ to $s_L$ in $\phi_0$.
By the properties of $\phi$ we infer that in this way we obtain a multiset $S'$ of exactly $k$ centers
and an assignment $\phi_0$, satisfying $\dist_G(v,\phi_0(v)) \le 11$ for each $v \in V$
and $|\phi^{-1}(s')| \le L(s')$ for each $s' \in S$.
\end{proof}

\begin{corollary}
The integrality gap of LP1 for connected graphs is at most $\softconst$
and there is a $\softconst$-approximation algorithm for connected graphs.
\end{corollary}

Since Theorem~\ref{thm:conn} shows that 
a $c$-approximation algorithm for connected graphs implies a
$c$-approximation algorithm for general graphs, the proof of 
Theorem~\ref{thm:soft-apx} follows.

\section{Uniform capacities}
\label{sec:uniform} 

Here, we prove Theorem~\ref{thm:gap-upper-constant}, i.e. we show $5$ and $6$ upper
bounds for the integrality gap of LP1 for uniform-soft-capacities
and uniform-hard-capacities respectively, which is a counterposition to 
lower bounds of $4$ and $5$ as stated in Theorem~\ref{thm:gap-constant}.

The algorithm of Khuller and Sussmann~\cite{KS} gives $6$- and $5$-approximation algorithm for the uniform capacitated $k$-center problem,
for hard and soft capacities respectively.
It is possible to reformulate this algorithm to make it a rounding algorithm for LP1.
However to avoid rewriting the whole algorithm, we show that if the algorithm of \cite{KS}
does not produce a solution (which means that there is no solution which uses at most one hop),
then it produces a witness, showing that there is no solution using $k$ vertices and covering within distance
of at most one hop.
For a given instance $(G,k,L)$ of uniform capacitated $k$-center problem if the algorithm of Khuller and Sussmann
does not produce a solution, both in the case of hard and soft capacities, then it creates a set $V_0 \subseteq V(G)$,
such that:
\begin{itemize}
  \item[(i)] for any pair of vertices $u,v \in V_0$, the distance between $u$ and $v$ in $G$ is at least $3$,
  \item[(ii)] $|V_0|+\frac{|V'|}{L} > k$, where $V' = \{v\in V(G): \forall_{u\in V_0} \dist_G(u,v) \ge 3\}$.
\end{itemize}
Using the properties (i) and (ii) of $V_0$ we prove that there is no feasible solution for LP1.
By (i) and by constraints (\ref{con4}) and (\ref{con6}) of LP1 we have:
\begin{align}
\label{ineq4}
|V_0| = \sum_{v \in V_0} \sum_{u \in N_G[v]} x_{u,v} \le \sum_{v \in V_0} \sum_{u \in N_G[v]} y_u = \sum_{v \in N_G[V_0]} y_v\,,
\end{align}
where the last equality follows from (i).
We lower bound the sum of $y$-values of vertices of $N_G[V']=V \setminus N_G[V_0]$ as follows.
\begin{align}
\label{ineq5}
\frac{|V'|}{L} = \sum_{v \in V'} \sum_{u \in N_G[v]} \frac{x_{u,v}}{L} \le \sum_{u \in N_G[V']} y_u
\end{align}
The first equality follows from constraint (\ref{con4}) and (\ref{con6}) of LP1, whereas
the inequality follows from (\ref{con3}) of LP1.
Therefore by (\ref{ineq4}), (\ref{ineq5}) and (ii), we infer constraint (\ref{con1}) of LP1 is violated.
\begin{align*}
\sum_{v \in V}y_v = \sum_{v\in N_G[V']}y_v + \sum_{v \in N_G[V_0]}y_v \ge |V_0|+\frac{|V'|}{L} > k
\end{align*}
Hence there is no feasible solution for LP1, which finishes the proof of Theorem~\ref{thm:gap-upper-constant}.
\qed

\section{Integrality gap lower bounds}
\label{sec:gaps}

In this section we present lower bounds on the integrality gap of LP1, i.e. prove Theorem~\ref{thm:gap-constant}.

We start with a construction for uniform hard capacities.
Let $k \ge 24$ be an integer and set the uniform capacity as $L=k-1$.
Let $G'$ be a graph which consists of two adjacent vertices $a,b$ together with $L+2$ vertices adjacent to both $a$ and $b$.
Let $G$ be a graph composed of:
\begin{itemize}
  \item $k-6$ copies of the graph $G'$ denoted as $G_i$ for $i=1,\ldots,k-6$,
  \item a single star $S_{k-6}$, rooted at $r$, which has exactly $k-6$ leaves denoted as $\ell_1,\ldots,\ell_{k-6}$,
  \item $k-6$ vertices $x_i$ for $i=1,\ldots,k-6$,
  \item $k-6$ edges $x_i,\ell_i$,
  \item for each $i=1,\ldots,k-6$ two edges between $a_i,b_i \in V(G_i)$ and $x_i$.
\end{itemize}

\begin{figure}[h]
\begin{center}
\includegraphics[scale=1]{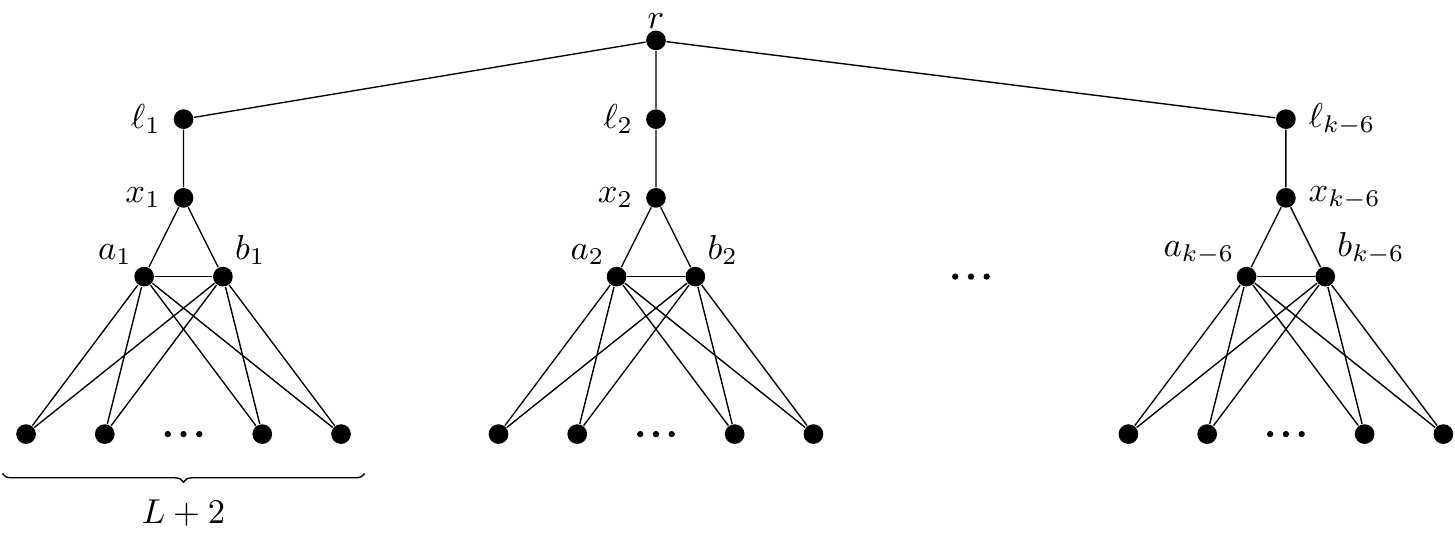}
\caption{The graph $G$ used in the proof of Theorem~\ref{thm:gap-constant}.
}
\label{fig6}
\end{center}
\end{figure}
Note that the graph $G$ is connected.
Observe that by setting the $y$-value equal to $1$ for the vertex $r$, and $y_{a_i}=y_{b_i}=\frac{L+5}{2L} \le 1$ for each $i=1,\ldots,k-6$,
we can set $x$ variables appropriately to obtain a feasible solution to the LP relaxation
since the sum of $y$ variables is equal to $1+2(k-6)\frac{L+5}{2L}=1+\frac{(k+4)(k-6)}{k-1} \le 1+\frac{(k-1)^2}{k-1} \le k$.

Let us assume that there is a hard-capacitated $k$-center in the graph $G^4$.
By $B_i$ (for $1 \le i \le k-6$) let us denote the set of vertices of $G_i$
of degree two.
Consider any $1 \le i \le k-6$, and let us focus on how vertices of $B_i$
can be covered.
Observe that the only vertices in $G$ that are within distance $4$
from $B_i$ are $V(G_i)\cup\{x_i,\ell_i, r\}$.
Since $|B_i| = L+2 > L$, at least two centers are opened in the set
$V(G_i)\cup\{x_i,\ell_i, r\}$, which mean
that at least one center is opened in $V(G_0)\cup \{x_i,\ell_i,r\}$.
Furthermore for different values of $i$, the sets $V(G_i)\cup\{x_i,\ell_i\}$
are disjoint.
Moreover if only one center is opened in the set $V(G_i) \cup \{x_i,\ell_i\}$,
then at least two vertices of $B_i$ have to be covered by $r$.
Hence even if a center is opened in $r$, then for at least $(k-6)-\frac{L}{2}$ values of $i$
we need to open at least two centers in the set $V(G_i)\cup\{x_i,\ell_i\}$.
Moreover we need to open at least one center in each of the sets $V(G_i)\cup\{x_i,\ell_i\}$, since otherwise $L+2$ vertices can not be covered.
In total there are at least $(k-6)+\big((k-6)-\frac{L}{2}\big)$ centers opened, but $(k-6)+\big((k-6)-\frac{L}{2}\big)=k+(k-12)-\frac{k-1}{2}=k+\frac{k-23}{2} > k$ since $k \ge 24$, 
a contradiction.

For the uniform-soft-capacitated case the analysis is even simpler since if $G^3$ admits a soft-capacitated $k$-center then 
without loss of generality there are at least two centers open in each of the sets $V(G_i)\cup\{x_i,\ell_i\}$.

To prove the lower bound for non-uniform capacities, in the above example we change capacities to zero
for all the vertices except $r$ and each of the $2(k-6)$ vertices $a_i$, $b_i$.
Note that those are the only vertices with non-zero $y$-value in the constructed feasible solution for LP1.
Moreover we have shown that in any integral solution there exists a vertex of $B_i$,
which has to be covered by a vertex of $V(G_j) \cup \{x_j,\ell_j\}$, for $j\not=i$,
and since the only vertices of non-zero capacity in this set are the vertices $a_j$, $b_j$,
we infer, that some vertex of $B_i$ has to be covered by a vertex at distance at least $7$.
Consequently the proof of Theorem~\ref{thm:gap-constant} follows. \qed

\section{$(3-\eps)$-approximation hardness}
\label{sec:lowerbound}

In this section we show, that it is not possible to approximate the $k$-center problem
with non-uniform capacities in polynomial time with approximation ratio $(3-\eps$) for constant $\eps > 0$.

\begin{theorem}
\label{thm:3-hardness}
If there exists a $(3-\eps)$-approximation algorithm for the $k$-center problem with 
non-uniform capacities, where all the non-zero capacities, then $P=NP$.
The theorem applies both to hard and soft capacitated version.
\end{theorem}

\begin{proof}
Let us assume that a $(3-\eps)$-approximation algorithm exists.
We present a reduction from the \exactcover problem, which is NP-complete.

\defproblemu{\exactcover}
{A set system $(\cF,U)$, where each set in $\cF$ has exactly $3$ elements.}
{Does there exist a subset $\cF' \subseteq \cF$, such 
that each element of the universe $U$ belongs to exactly one set in $\cF'$.}

Let $I=(\cF,U)$ be an instance of \exactcover.
As a graph $G$ we take the bipartite graph $(\cF \cup B, E)$,
where $B=\bigcup_{i=1}^{|\cF|+1} U_i$, $U_i=\{u_i: u \in U\}$
and $E=\{Su_i : S \in \cF \wedge 1 \le i \le |\cF|+1 \wedge u \in S\}$.
That is $G$ is an incidence graph of the set system $(\cF,U)$, 
where the universe is replicated $|\cF|+1$ times.
Additionally for each $S \in \cF$ we add to the graph a vertex $x_S$,
which is adjacent to $S$, and has exactly $3|\cF|+1$ pendant neighbors (see Fig.~\ref{fig7}).
We set a capacity function $L:V(G) \rightarrow \Z_+ \cup \{0\}$  as follows:
\begin{itemize}
  \item For each $u_i \in U_i$, $L(u_i) = 0$.
  \item For each $S \in \cF$, $L(S) = L(x_S) = 3|\cF|+3$.
  \item For each pendant vertex $v$, which is adjacent to $x_S$, for $S \in \cF$, we set $L(v)=0$.
\end{itemize}

\begin{figure}[h]
\begin{center}
\includegraphics[scale=0.68]{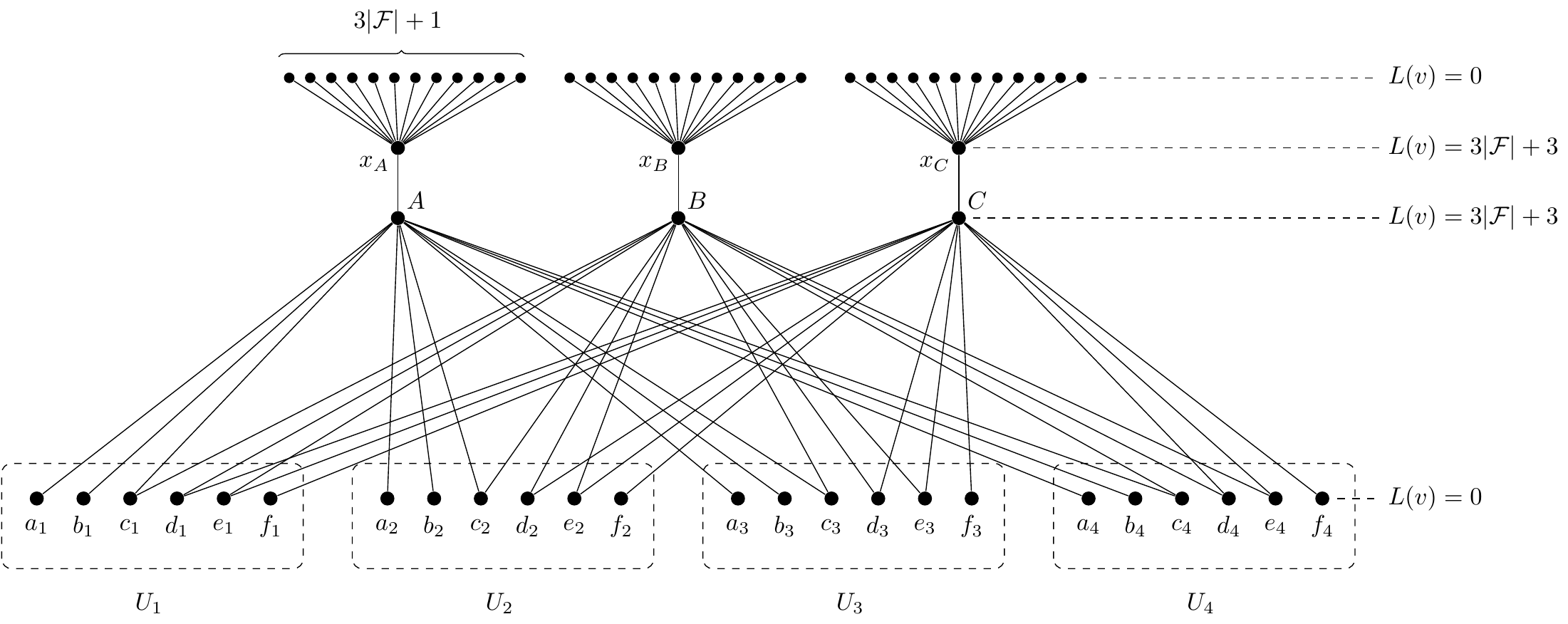}
\caption{The graph $G$ constructed for the set system $(\cF,U)$, where $\cF=\{\{A=\{a,b,c\},B=\{c,d,e\},C=\{d,e,f\}\}\}$ and $U=\{a,b,c,d,e,f\})$}
\label{fig7}
\end{center}
\end{figure}

Note that all the vertices, which have non-zero capacities, have exactly the same capacity.
Let $I'=(G,L,k=|\cF|+|U|/3)$ be an instance of the capacitated $k$-center problem.
In what follows we prove:
\begin{itemize}
  \item[(i)] If $I$ is a YES-instance, then there exists a set $V_0 \subseteq V(G)$, such that $|V_0|=k$,
  together with a function $\phi:V\rightarrow V_0$, which satisfies $\forall_{v\in V(G)} \dist_G(v,\phi(v)) \le 1$ and
  $\forall_{v \in V_0} |\phi^{-1}(v)| \le L(v)$.
  Less formally, $V_0$ is a solution for the capacitated $k$-center problem,
  which opens at most one center in each vertex.
  \item[(ii)] Let $V_0$ be a multiset containing exactly $k$ vertices of $V(G)$,
  such that there exists a function $\phi:V\rightarrow V_0$, which satisfies $\forall_{v\in V(G)} \dist_G(v,\phi(v)) < 3$ and
  $\forall_{v \in V_0} |\phi^{-1}(v)| \le L(v)$.
  Then one can in polynomial time construct a set $\cF' \subseteq \cF$, such that each element of the universe $U$
  belongs to exactly one set in $\cF'$.
  Intuitively, given a solution to the capacitated $k$-center problem which uses distances at most two, and potentially
  opens more than one center in a vertex, one can in polynomial time construct a solution for the instance $I$.
\end{itemize}
Observe that having (i) and (ii) suffices to prove the theorem, since if $I$ is a YES-instance,
we can obtain a solution to $I$, by constructing the graph $G$, running the $(3-\eps)$-approximation algorithm on $(G,L,k)$,
which by (i) returns a set $V_0$, which by (ii) we transform to a solution for $I$ in polynomial time.

First we prove (i).
Let $\cF'$ be a solution for the instance $I$.
We take $V_0=\cF' \cup \{x_S : S\in \cF\}$, that is $|U|/3$ sets from $\cF$, and all the vertices $x_S$.
Note that $|V_0| = |U|/3+|\cF|$ and each vertex in $V_0$ has capacity exactly $3|\cF|+3$.
Observe that if each vertex $S \in \cF'$ covers each of the $|\cF|+1$ copies of the three elements in $S$,
while each vertex $x_S$ covers itself, all the $3|\cF|+1$ pendant vertices, and the vertex $S$,
then we obtain the desired function $\phi$.

Next we prove (ii).
Observe, that $|V(G)|=(|\cF|+1)|U|+|\cF|(3|\cF|+3)=(3|\cF|+3)(|U|/3+|\cF|)$,
and hence in the multiset $V_0$ there are only vertices of capacity $3|\cF|+3$,
since otherwise the total capacity of the vertices in $V_0$ would be smaller than $|V(G)|$.
Moreover, without loss of generality we can assume that the multiset $V_0$
contains only vertices $S \in \cF$, since we can always replace a vertex $x_S$ by
$S$ without exceeding distance two, in the function $\phi$.
Furthermore pendant vertices of any $x_S$ are covered by vertices at distance at most two,
that is by the vertex $S$.
Therefore each $S \in \cF$ appears in the multiset $V_0$ at least once.
Let $\cF' \subseteq \cF$ contain exactly those sets $S \in \cF$,
for which the vertex $S$ belongs to $V_0$ more than once.
Note that since $|V_0|=|\cF|+|U|/3$, the set $\cF'$ contains at most $|U|/3$ sets.
Consequently to prove (ii) it is enough to show that each element of the universe $U$
belongs to at least one set in $\cF'$.

Let $V_0' \subseteq V_0$ be the set containing the first copy of each
$S \in \cF$ in $V_0$, namely $V_0'=\cF$.
Without loss of generality, for each $S \in V_0'$ the set $\phi^{-1}(S)$
contains $x_S$ and all its pendant neighbors, since none
of those vertices can be covered by $S' \not= S$, because the
distance would be at least three.
Consequently for each $S \in V_0'$ we have $|\phi^{-1}(S) \cap (\bigcup_{i=1}^{|\cF|+1} U_i)| \le L(S)-(3|\cF|+2)=1$.
Since $|V_0'| \le |\cF| < |\cF|+1$, there exists an index $1 \le i_0 \le |\cF|+1$,
such that for each $S \in V_0'$ we have $\phi^{-1}(S) \cap U_i = \emptyset$.
Therefore each vertex $u_i \in U_i$ is covered by some set $S \in \cF'$ that it belongs to,
which proves (ii) and finishes the proof of Theorem~\ref{thm:3-hardness}.

\end{proof}

\section{Conclusions and open problems}

We have obtained the first constant approximation ratio for the $k$-center 
problem with non-uniform hard capacities.
The approximation ratio we obtain is in the order of hundreds (however we do not calculate it explicitly), 
so the natural open problem is to give an algorithm with a reasonable approximation ratio.
Moreover, we have shown that the integrality gap of the standard LP formulation for connected graphs 
in the uniform capacities case is either 5 or 6, which we think might be an evidence, 
that it should be possible to narrow the gap between the known lower bound of $(2-eps)$ and upper bound $6$ 
in the uniform capacities case.

\section*{Acknowledgements}

We are thankful to anonymous referees for their helpful comments
and remarks.

\bibliographystyle{abbrv}
\bibliography{kcenter}
\newpage

\end{document}